%% file: main_arxiv.tex
\documentclass[11pt]{article}

\usepackage{enumitem}
\usepackage{setspace}
\usepackage[ruled,vlined]{algorithm2e}
\usepackage{amsmath,amssymb,amsfonts}%
\usepackage{amsthm}%
\usepackage{thm-restate}
\usepackage{mathtools}
\usepackage[title]{appendix}%
\usepackage[dvipsnames]{xcolor}%
\usepackage{hyperref}
\usepackage{cleveref}%
\usepackage{environ}
\usepackage{tabularx}
\usepackage{csquotes}
\usepackage{standalone}
\usepackage{caption}
\usepackage{tikz}
\usepackage{subcaption}
\usepackage{geometry}
\graphicspath{{tikz/},{}}

\usetikzlibrary{decorations.pathreplacing}
\usepackage{tikz-network}
\usepackage{breakurl}
\usepackage[
backend=biber,
style=apa,
]{biblatex}

\usepackage{xpatch}

\makeatletter
\patchcmd\blx@bblinput{\blx@blxinit}
                      {\blx@blxinit
                      }{}{\fail}
\makeatother
\newtheorem{theorem}{Theorem}

\newtheorem{proposition}[theorem]{Proposition}
\newtheorem{corollary}[theorem]{Corollary}
\newtheorem*{claim}{Claim}
\newtheorem{lemma}[theorem]{Lemma}

\theoremstyle{remark}

\newtheorem*{remark}{Remark}

\theoremstyle{definition}
\newtheorem{definition}[theorem]{Definition}

\crefname{lemma}{lemma}{lemmata}
\crefname{property}{property}{properties}
\crefname{condition}{condition}{conditions}
\crefname{equation}{Eq.}{Eqs.}
\crefname{enumi}{property}{properties}
\Crefname{enumi}{Property}{Properties}
\Crefname{proposition}{Proposition}{Propositions}

\newcommand{\lr}[1]{\left(#1\right)}
\newcommand{\lre}[1]{\left[#1\right]}
\newcommand{\lrb}[1]{\left\lbrace#1\right\rbrace}

\newcommand{\Oless}[1]{\mathcal{O}\lr{#1}}

\newcommand{\Oequal}[1]{\Theta\lr{#1}}

\newcommand{\runtime}[1]{\Oless{#1}}

\newcommand{\lw}{l_w}
\newcommand{\lo}{l_o}
\newcommand{\uw}{u_w}
\newcommand{\uo}{u_o}
\newcommand{\Uw}{U_w}
\newcommand{\Uo}{U_o}
\newcommand{\Startw}[1]{S_w^{#1}}
\newcommand{\Termw}[1]{T_w^{#1}}
\newcommand{\Starto}[1]{S_o^{#1}}
\newcommand{\Termo}[1]{T_o^{#1}}

\newcommand{\request}[1]{r^{#1}}
\newcommand{\requestl}[1]{r_l^{#1}}
\newcommand{\requestu}[1]{r_u^{#1}}
\newcommand{\Request}[1]{R^{#1}}

\newcommand{\work}[1]{w^{#1}}
\newcommand{\Work}[1]{W^{#1}}

\newcommand{\DODOSP}{\textsc{Days On Days Off Scheduling Problem}\xspace}
\newcommand{\LDODOSP}{\textsc{Local Days On Days Off Scheduling Problem}\xspace}
\newcommand{\UDODOSP}{\textsc{Only Upper Bounded Days On Days Off Scheduling Problem}\xspace}
\newcommand{\dodosp}{\textsc{DODOSP}\xspace}
\newcommand{\ldodosp}{\textsc{LDODOSP}\xspace}
\newcommand{\cdodosp}{\textsc{Cyclic DODOSP}\xspace}
\newcommand{\udodosp}{\textsc{UDODOSP}\xspace}
\newcommand{\threepartition}{\textsc{$3$-Partition Problem}\xspace}
\newcommand{\rthreepartition}{\textsc{Restricted $3$-Partition Problem}\xspace}
\newcommand{\CP}{\textsc{Circulant Problem}\xspace}

\newcommand{\WORK}{\text{ON}}
\newcommand{\OFF}{\text{OFF}}
\newcommand{\TASK}{\text{SHIFT}}

\newcommand{\N}{\mathbb{N}}

\newcommand{\IN}{\mathbb{N}}
\newcommand{\CO}{\mathcal{O}}
\newcommand{\CD}{\mathcal{D}}
\newcommand{\CN}{\mathcal{N}}

\makeatletter
\newcommand{\problemtitle}[1]{\gdef\@problemtitle{#1}}%
\newcommand{\probleminput}[1]{\gdef\@probleminput{#1}}%
\newcommand{\problemquestion}[1]{\gdef\@problemquestion{#1}}%
\newcommand{\problemsep}[1]{\gdef\@problemsep{#1}}
\NewEnviron{problem}{
  \problemtitle{}\probleminput{}\problemquestion{}%
  \BODY%
  \par\addvspace{.5\baselineskip}
  \centering
  \noindent \fbox{
  \begin{tabularx}{0.85\textwidth}{@{\hspace{\parindent}} l X c}
    \multicolumn{2}{@{\hspace{\parindent}}l}{\@problemtitle} \\[7pt]%
    \textbf{Input:} & \@probleminput%
    
    \textbf{Question:} & \@problemquestion%
  \end{tabularx}}
  \par\addvspace{.5\baselineskip}
}
\newcommand{\customlabel}[2]{%
\protected@write\@auxout{}{\string\newlabel{#1}{{#2}{}}}}

\newenvironment{claimproof}{%
  \proof[Proof of Claim]}{\endproof}

\def\namedlabel#1#2#3{\begingroup
    #2%
    \def\@currentlabel{#2}%
    \phantomsection\label[#3]{#1}\endgroup
}
\makeatother

\providecommand{\keywords}[1]
{
  {\small	
  \textbf{\textit{Keywords---}} #1}
}
\usepackage[dvipsnames]{xcolor}
\usepackage{tikz}
\usepackage{intcalc}
\usetikzlibrary{calc}
\usetikzlibrary{patterns}
\usetikzlibrary{decorations.pathreplacing,calligraphy}
\usetikzlibrary{decorations.pathmorphing}
\usetikzlibrary{arrows}
\usetikzlibrary{fadings}

\usepackage{tikz-network}

\newcommand{\offsquare}[2]{
    \draw[] (#1) rectangle (#2);
}
\newcommand{\undecidedsquare}[2]{
    \node[] at ($(#1)!0.5!(#2)$) {?};
    \draw[] (#1) rectangle (#2);
}
\newcommand{\worksquare}[2]{
    \draw[] (#1) rectangle (#2);
    \fill[gray] (#1) rectangle (#2);
}

\definecolor{col_highlighted}{HTML}{31782F}

\newcommand{\nurseplan}[4]{
\begin{scope}[shift = {(0,-#1)}]
    \foreach \x in {#2}{
        \worksquare{\x,0}{\x+1,1}
    }
    \foreach \x in {#3}{
        \offsquare{\x,0}{\x+1,1}
    }
    \node at (-0.7,0.5) {#4};
\end{scope}
}
\newcounter{numdays}
\newcommand{\demandvalues}[1]{
\setcounter{numdays}{0}
\begin{scope}[shift = {(0,0)}]
    \foreach \x in {#1}{
        \node at (\value{numdays}+0.5,0.5) {$\x$};
        \addtocounter{numdays}{1}
    }
    \draw (0,0) -- (\value{numdays},0);
        \foreach \x in {0,...,\value{numdays}}{
            \draw (\x,-0.15) -- (\x,0.15);
        }
        \node at (-0.7,0.5) {$r^d$};
\end{scope}
}
\newcommand{\periodcounter}[3]{
\setcounter{numdays}{0}
\begin{scope}[shift = {(0,-#2)}]
    \foreach \x in {#1}{
        \node at (\value{numdays}+0.5,0.3) {$\x$};
        \addtocounter{numdays}{1}
    }
    \draw (0,0) -- (\value{numdays},0);
        \foreach \x in {0,...,\value{numdays}}{
            \draw (\x,-0.1) -- (\x,0.1);
        }
        \node at (-0.7,0.3) {#3};
\end{scope}
}

\newcommand{\periodstart}[1]{
    \coordinate (a) at (#1);
   \draw[fill=green!60!black,draw=green!50!black] (a) -- ($(a)+(0.5,-0.5)$) -- ($(a)+(0,-1)$) -- cycle;
}
\newcounter{nurseID}
\newcommand{\orderedperiodstarts}[2]{
\setcounter{nurseID}{0}
\begin{scope}[shift = {(0,0)}]
    \foreach \x in {#2}{
        \periodstart{\x,-\value{nurseID}}
        \addtocounter{nurseID}{1}
        \setcounter{nurseID}{\intcalcMod{\value{nurseID}}{#1}}
    }
\end{scope}
}
\newcommand{\periodend}[1]{
    \coordinate (a) at (#1);
   \draw[fill=red,draw=red] ($(a)+(1,0)$) -- ($(a)+(0.5,-0.5)$) -- ($(a)+(1,-1)$) -- cycle;
}
\newcommand{\orderedperiodends}[2]{
\setcounter{nurseID}{0}
\begin{scope}[shift = {(0,0)}]
    \foreach \x in {#2}{
        \periodend{\x,-\value{nurseID}}
        \addtocounter{nurseID}{1}
        \setcounter{nurseID}{\intcalcMod{\value{nurseID}}{#1}}
    }
\end{scope}
}

\def\vdis{0.48}
\ifcsname c@mycounter\endcsname
\else
  \newcounter{mycounter}
\fi
\newcommand{\add}[2]{
    \foreach \x in {1,...,#2}{
        \node at (\vdis*\themycounter,0) {$#1$};
        \stepcounter{mycounter}
    }
}

    \renewcommand{\add}[1]{
        \node at (\vdis*\themycounter,0) {$#1$ };    
        \stepcounter{mycounter}
    }
    \newcommand{\addtriple}[1]{
        \add{#1}
        \add{\dots}
        \add{#1}
    }
    \newcommand{\addbox}[1]{
        \draw (\vdis*\themycounter-\vdis*0.45,-0.3) rectangle (\vdis*\themycounter+\vdis*1.45,0.3);
        \node[] at (\vdis*\themycounter+\vdis*0.5,0) {$l_w\times#1$ };
        \stepcounter{mycounter}
        \stepcounter{mycounter}
    }

\tikzstyle{circlenode}   = [circle,minimum size=4mm,align=center,text width=5mm,
											  	inner sep=0.8mm,outer sep=0.5mm,draw,draw=black, text = blue, thick]
\tikzstyle{varnode}  = [circle,minimum size=4mm,align=center,text width=4mm,
											  	inner sep=0pt,outer sep=2.2mm, fill opacity = 0, text opacity = 1, text = blue]
\tikzstyle{edgenode}   = [fill=white,circle,inner sep=0pt,minimum size=4mm]
\tikzstyle{normaledge} = [->,>={Stealth}, very thick, black]

\graphicspath{{Tikz},{Paper}}

\newcommand{\mailto}[1]{\href{mailto:#1}{#1}}

\title{The \DODOSP}
\author{Fabien Nießen\thanks{Department of Computer Science, KU Leuven, Gebroeders De Smetstraat 1, 9000 Gent, Belgium ({\mailto{fabien.niessen@kuleuven.be}}), supported by KU Leuven C24E/23/012 'Human-centred decision support based on new theory for personnel rostering', corresponding author}
\and Paul Paschmanns\thanks{Research Institute for Discrete Mathematics, University of Bonn, Lenn\'estr.~2, 53113 Bonn, Germany ({\mailto{paschmanns@dm.uni-bonn.de}})
}
}
\date{}

\makeatletter
\renewcommand{\@fnsymbol}[1]{\ifcase#1\or a\or b\or c\or d\or e\or f\or *\else \@ctrerr\fi}
\makeatother

\begin{document}

\maketitle
\begin{abstract}
Personnel scheduling problems have received considerable academic attention due to their relevance in various real-world applications. These problems involve preparing feasible schedules for an organization's employees and often account for factors such as qualifications of workers and holiday requests, resulting in complex constraints.
While certain versions of the personnel rostering problem are widely acknowledged as NP-hard, there is limited theoretical analysis specific to many of its variants. Many studies simply assert the NP-hardness of the general problem without investigating whether the specific cases they address inherit this computational complexity. 

In this paper, we examine a variant of the personnel scheduling problems, which involves scheduling a homogeneous workforce subject to constraints concerning both the total number and the number of consecutive work days and days off.
This problem was previously claimed to be NP-complete. In this paper, we prove its NP-completeness and investigate how the combination of constraints contributes to this complexity. Furthermore, we analyze various special cases that arise from the omission of certain parameters, classifying them as either NP-complete or polynomial-time solvable. For the latter, we provide easy-to-implement and efficient algorithms to not only determine feasibility, but also compute a corresponding schedule.
\end{abstract}


\clearpage

\section{Introduction}
In personnel scheduling problems the task is to compute a schedule for a given set of workers subject to certain constraints, such as daily requests on the number of active workers or an upper limit on the number of consecutive night shifts, possibly aiming to optimize a given objective function.
Such problems are well-studied from a practical point of view, due to its many applications. 

Given the numerous constraints involved, personnel rostering problems are frequently addressed using integer linear programming (ILP) or LP-based heuristics \parencite{burkeState2004}. \textcite{burkeState2004} also provide a comprehensive overview of nurse rostering -- one of the most widely studied applications within personnel scheduling. Efforts to categorize nurse rostering problems in a manner similar to scheduling are presented by \textcite{causmaeckerCategorisation2011}. For examples of modeling nurse rostering  see \parencite{glassNurse2010,smetNurse2013,smetModelling2014}. Regarding the compuational complexity, the NP-hardness of specific cases within personnel rostering has been established in several studies, including key contributions by \textcite{lauComplexity1996}, \textcite{osogamiClassification2000}, and \textcite{hartogComplexity2023}. In contrast, polynomially solvable cases -- where particular constraints or simplified structures allow for efficient solutions -- are for example examined by \textcite{bruckerPersonnel2011} and \textcite{smetPolynomially2016}. A comprehensive summary covering both NP-hard and polynomially solvable cases is available in the thesis by \textcite{hartogComplexity2016}.\medskip

In this paper we are focusing on the decision variant of the \DODOSP (\dodosp), a personnel scheduling problem involving a homogeneous set of workers, introduced by \textcite{brunnerBoundedFlexibility2013}. The \dodosp involves a sequence of days with minimum requests for the number of workers needed each day, along with the following parameters: upper limits on the total number of work days and days off for each worker, as well as upper and lower bounds on the number of consecutive work days and days off.

These constraints are derived from practical considerations. For example, an upper limit on consecutive working days could stem from ensuring compliance with labor regulations. 
For example, in Germany, exceeding 19 consecutive days of work violates the \enquote{Arbeitszeitgesetz} \parencite{arbeitszeitgesetz}.
Similarly, the total number of working days of each worker during the given time horizon might be limited.
On the other hand, employees usually only have a limited number of vacation days. A lower bound on the number of consecutive work days or days off can prevent wasting too much time to commute and might therefore be especially interesting for industries that involve visiting customers on site.

While \textcite{brunnerBoundedFlexibility2013} primarily focused on practical results using a branch-and-price approach, they claimed the NP-completeness of the \dodosp. However, the proof they provided was unfortunately incorrect.
In this paper, we prove that the problem is indeed NP-complete. The complexity primarily arises from combining lower bounds on consecutive work days or days off with an upper bound on each worker's total number of work days or days off. Additionally, we classify each special case that results from omitting certain constraints as either polynomial-time solvable or NP-complete. All hardness results are derived from reductions from the \threepartition. For the polynomial-time solvable cases, we reduce the \dodosp to the problem of checking for negative cycles in a directed graph.\medskip 

The remainder of this paper is organized as follows. We formally introduce the \dodosp and the special cases we study in \Cref{sec: problem intro}. Furthermore, we outline the main ideas used throughout the paper. The hardness results are given in \Cref{sec: np completeness}. Then, in \Cref{sec: local bounds,sec: upper bounds}, we provide polynomial-time algorithms for the special cases that are not NP-complete. In  \Cref{sec: optimization}, we briefly discuss how the results can be used in case of an optimization variant of the problem. Finally, \Cref{sec: summary} concludes the paper by summarizing our results and providing guidance on how they can be used, while also identifying open questions that remain for future research.
\section{Problem introduction}\label{sec: problem intro}
This chapter begins by first providing a formal definition of the \DODOSP (\dodosp) in \Cref{subsec: prob introd - formal definition} and introduces the terminology needed to understand the problem. Furthermore, we establish a set of special cases of the \dodosp, which will come in handy to analyze where the hardness originates from. Finally, in \Cref{subsec: prob introd - overview} we present the main concepts used throughout
the paper.

\subsection{Formal definition}\label{subsec: prob introd - formal definition}
In this paper we are always given a number of workers $N\in\IN$ and a number of days $D\in \IN$. We then denote the set of workers by $\CN \coloneqq \lrb{n_1,\dots,n_N}$ and the set of days by $\CD \coloneqq \lre{D} = \lrb{1,\dots,D}$.
\begin{definition}
    Given a number of workers $N$ and a number of days $D$, a mapping $f : \CN\times \CD \to \lrb{\WORK,\OFF}$ is called a \textit{schedule}. We say worker $n_i$ works or is on duty on day $d\in \CD$ if $f\lr{n_i,d}=\WORK$. Otherwise, we say that worker $n_i$ has day $d$ off. 
\end{definition}

\begin{definition}
    Given a worker $n_i$, we call every inclusion-wise maximal set of consecutive days on which $n_i$ works a \textit{work period} (of worker $n_i$). Similarly, we define an \textit{off period} (of worker $n_i$) as an inclusion-wise maximal set of days on which $n_i$ does not work.
\end{definition}
\noindent \Cref{fig: feasible example} provides an example schedule and illustrates the concept of work periods and off periods.
We define the \DODOSP as the following decision problem:
\begin{problem}
    \problemtitle{\DODOSP (\dodosp)}
    \probleminput{A number of days $D \in \IN$, a number of workers $N\in \IN$, bounds $\lw, \uw, \lo, \uo, \Uw, \Uo \in \IN$ on periods and bounds $0\leq\requestl{d}\leq\requestu{d}\leq N$ for each $d\in \CD$ on requests for all days.\\[1cm]}
    \problemsep{5pt}
    \problemquestion{Does there exist a \emph{feasible} schedule with $N$ workers? A schedule is called feasible if on every day $d$ at least $r_l^d$ and at most $r_u^d$ workers work and the bounds for every worker are respected, i.e.,
    \begin{enumerate}[before=\vspace{-0.5\baselineskip},after=\vspace{-\baselineskip},label=\textbullet]
        \item every work period of every worker is at least $l_w$ and at most $u_w$ days long,
        \item every off period of every worker is at least $l_o$ and at most $u_o$ days long,
        \item every worker works in total at most $U_w$ days and
        \item every worker has in total at most $U_o$ many days off.
    \end{enumerate}}
\end{problem}
The \dodosp was first introduced by \textcite{brunnerBoundedFlexibility2013}, who proposed a branch-and-price algorithm to solve real-world instances optimally. We have expanded their problem definition by adding an upper bound $\Uo$ on the total number of days off as well as day-specific upper bounds $\requestu{d}$ on the number of workers that should be scheduled to work on the given day.
\clearpage

We distinguish the problem's bounds in the following way: We call $l_w$ and $l_o$ the \emph{local lower bounds}, $u_w$ and $u_o$ the \emph{local upper bounds}, $U_w$ and $U_o$ the \emph{global (upper) bounds} and $\requestl{d}$ and $\requestu{d}$ the request bounds.

\begin{remark}
    Note that the definition of the \dodosp does not introduce global lower bounds. This is because they can easily be transformed into global upper bounds of the opposite type: A global lower bound of $L_w$ on the work days is equivalent to a second upper bound of $D-L_w$ concerning the number of days off. Thus, instead of using $L_w$, we would use $\min\lrb{U_o, D-L_w}$ as a global upper bound on the number of days off. The same argument holds for a global lower bound $L_o$ for the number of days off.
    Further, we point out that the schedule is not cyclic, i.e., is not repeated after $D$ days. Thus, a worker who works on the first day has to work for the first $\lw$ days.
    Nevertheless, all our hardness results from \Cref{sec: np completeness} can be extended to the \cdodosp, where we assume that the schedule is repeated. This affects the impact of local bounds, i.e., it would be feasible to work for only two days at the beginning of the schedule and $\lw-2$ days at the very end of the schedule.
\end{remark}%
\begin{figure}[htbp]
    \centering
    \includestandalone[ width=0.85\textwidth]{prob_intro_feasible_example}
    \caption{A feasible solution for the instance of the \dodosp  with exact demands (i.e., $\request{d}\coloneqq\requestl{d}=\requestu{d}$ for all $d$) given in the first row, work period bounds of $\lw=3$ and $\uw=5$, off period bounds of $\lo=2$ and $\uo=4$, as well as global bounds $\Uw=10$ and $\Uo=9$. Empty squares indicate a day off, while filled squares indicate a day on duty. }
    \label{fig: feasible example}
\end{figure}%

Next, we introduce two special cases which arise from ignoring a set of bounds. Note that instead of removing a bound from the problem definition we can fix it to the value given in \Cref{tab: trivial bounds}, which makes it trivially fulfilled. 

\begin{definition}
    We call the \dodosp restricted to instances with $\lw=\lo=1$ the \UDODOSP (\udodosp).
\end{definition}
\begin{definition}
    We call the \dodosp restricted to instances with $\Uw=\Uo=D$ the \LDODOSP (\ldodosp).
\end{definition}

We are interested in these two special cases since they can be decided in polynomial time, which we prove in \Cref{sec: upper bounds,sec: local bounds}. On the other hand, in \Cref{thm: NP-hard in general}, we show that the \dodosp is NP-complete as soon as it features both global bounds and local lower bounds.

In \Cref{sec: np completeness} we consider further special cases that arise from removing a given set of bounds. 

\begin{table}[htbp]
    \centering
    \renewcommand{\arraystretch}{1.4}
    \begin{tabular}{lcccccc c cc}
       bound & $l_w$ & $u_w$ & $l_o$ & $u_o$ & $U_w$ & $U_o$ & & $\requestl{}$ & $\requestu{}$\\\hline
       default value & 1 & $D$ & 1 & $D$ & $D$ & $D$ & & $0$ & $N$
    \end{tabular}
    \caption{Values to set bounds to in order to make them trivially fulfilled.}%
    \label{tab: trivial bounds}
\end{table}

From an application perspective, one might be interested in the case where $\requestu{d}=N$ for all $1\leq d\leq D$. In this case we only have a minimal work request that must be covered on each day. Besides a trivial exception, the same complexity results hold as in the general case where requests are bounded from both sides (cf. \Cref{tab: NP-hardness cases if request bounded by one side}). If the minimum and maximum requests are identical each day, we refer to the instance as having \emph{exact requests}. Analyzing the \udodosp with exact requests is simpler than doing so with general requests. Therefore, in \Cref{sec: upper bounds}, we will begin by examining the case with exact requests, followed by an analysis of the general requests case.

\subsection{Overview of main arguments}\label{subsec: prob introd - overview}

In \Cref{sec: upper bounds}, we study the \udodosp. To motivate why this special case might be easier to solve than the general \dodosp we first take a more heuristical look at it by considering the example in \Cref{fig: motivate udodosp}, where we are given an instance of the \dodosp with only two workers and exact requests ($\requestl{}=\requestu{}$). For days with a request of zero or two there is no decision to be made when designing a schedule. However, on days with request one we have to choose which worker will be on duty. For the example shown in \Cref{fig: motivate udodosp}, we are in the situation that before a certain day both workers were on duty for the same number of days. On the given day only the first worker $n_1$ is on duty. For the following day, one worker is requested and we have to decide who will be on duty.

\begin{figure}[htbp]
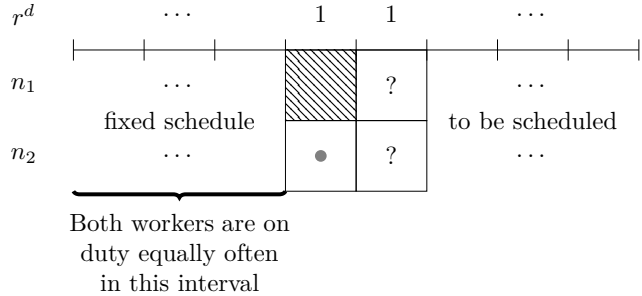

    \centering
    \includestandalone[mode=buildnew, subpreambles=true, width=0.55\textwidth]{motivate_udodosp}
    \caption{Assume the left part of the schedule is fixed already and we have to decide who is on duty on the following day. Given that on the first days both workers were on duty equally often all upper bounds ($\uw$, $\uo$, $\Uw$ and $\Uo$) \enquote{would prefer} that we switch who is on duty, while lower bounds \enquote{would prefer} us to keep $n_1$ working.}
    \label{fig: motivate udodosp}
\end{figure}

From the perspective of the global bound it seems obvious what to do: it would be better to assign $n_2$ to work on the following day since this equalizes the number of days off and days on duty between the workers.
From the perspective of the local upper bounds, it also seems to be beneficial to assign $n_2$ to work. In doing so we interrupt the work resp.\ off periods of both workers which gives us a better chance to not violate the local upper bounds.
However, from the perspective of local lower bounds it seems more reasonable to keep $n_1$ on duty for another day since otherwise we might end a work or off period before it reached its minimal length.
Thus, all upper bounds are \enquote{more likely} to be satisfied if we change who is on duty more often, while the lower bounds are \enquote{more likely} to be fulfilled if we change who is on duty less often.

This shared preference of all upper bounds for many changes in who is on duty will be turned into an efficient algorithm (for the case of exact requests) in \Cref{sec: upper bounds explicit requests} by simply assigning workers in a cyclic manner. Besides an algorithm that computes feasible schedules (if they exist), we also obtain feasibility criteria for the case of exact requests that are easy to check. Using these criteria we then generalize the result in \Cref{sec: upper bounds request intervals} to the \udodosp with daily upper and lower bounds $\requestu{}$ and $\requestl{}$ on the number of workers on duty. This is done in two stages by first choosing the number of workers that are on duty from the interval $\lre{\requestl{d},\requestu{d}}$ for each day $d$ in a way that preserves feasibility. This can be done by computing a feasible potential in a certain auxiliary graph. In the second stage we are left with the case of exact requests.

Since the results of \Cref{sec: upper bounds} heavily rely on maximizing the number of times a worker switches between days off and being on duty, they can not be generalized to also fulfill lower bounds. In fact the combination of local lower bounds with global upper bounds is what makes the \dodosp NP-hard. We prove this in \Cref{sec: np completeness} by reducing a packing problem to the \dodosp. To do so, we interpret the different workers as containers into which we have to pack given numbers. In this setting, a global upper bound corresponds to the capacity of each container. The numbers that have to be packed are encoded in a unary way in the sequence of requests. Thus one can think of a number to be in a container whenever the corresponding worker is on duty during the days that encode the corresponding number. The local lower bounds are needed to ensure that a number is not split between two containers or, equivalently, that only one worker is on duty during the encoding of that number. 

\begin{figure}[htbp]
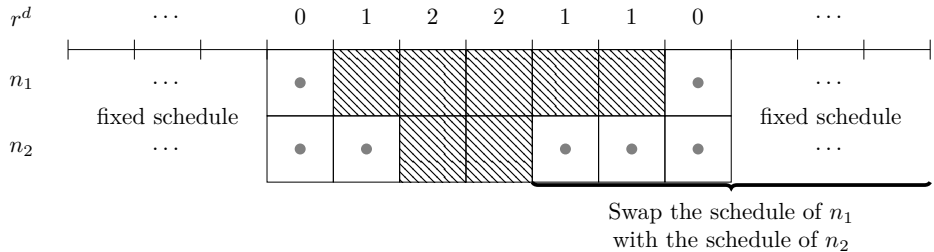

    \centering
    \includestandalone[mode=buildnew, subpreambles=true, width=0.8\textwidth]{motivate_ldodosp}
    \caption{The depicted schedule does not satisfy the FIFO property since the work period of $n_2$ starts later and finishes earlier than the one of $n_1$. If we would swap the suffixes of the schedules of $n_1$ and $n_2$, the FIFO property would hold (in the given part of the schedule). Note that swapping suffixes preserves feasibility with respect to local bounds.}
    \label{fig: motivate ldodosp}
\end{figure}

As described earlier, local lower bounds and local upper bounds have contradicting \enquote{preferences} when it comes to the number of times workers switch between being on duty and days off. However, if we would know for every day how many workers switch from being on duty to a day off and vice versa, it would be easy to construct a feasible schedule. This is because all local bounds benefit from a certain structure in the work and off periods of a schedule: whenever worker $n_1$ starts a work period earlier than another worker $n_2$ starts their work period, the work period of $n_1$ should not end later than the work period of $n_2$. The same should hold for off periods. We refer to this property as first-in-first-out (FIFO) property. \Cref{fig: motivate ldodosp} illustrates a situation that is avoided by the FIFO property. It also shows how to modify a feasible schedule in order to obtain the property. This modification preserves feasibility regarding local bounds, but might destroy feasibility regarding global bounds.

Since the FIFO property only distributes the number of work periods equally but does not take into account how long each work period is, it gives no guarantees in terms of global bounds (cf. \Cref{fig: no fifo instance}). In fact, even in schedules that satisfy the FIFO property, it can occur that the number of days on which two different workers are on duty, can differ by an arbitrary constant independent of $\lw$ and $\lo$. \Cref{fig: diff instance} illustrates what such an instance might look like.

\begin{figure}[htbp]
        \centering
    \begin{subfigure}{0.85\textwidth}
        \centering
        \includestandalone[mode=buildnew, subpreambles=true, width=0.8\textwidth]{diff_instance}
        \caption{Given $\lw = \lo = 2$, $\uw\geq3$, $\uo\geq3$ and $\Uw=\Uo=D=10$, the depicted schedule is the only solution of the instance with the given demands (up to symmetry in $n_1$ and $n_2$).
    If we expand the instance by repeating the pattern $0,1,2,1,1$ in the requests $c-2$ times, we receive an instance where one of the workers has to be on duty 
    $3c$ times, while the other works only on $2c$ days.}
    \label{fig: diff instance}
    \end{subfigure}
    \begin{subfigure}{0.85\textwidth}
        \centering
        \includestandalone[mode=buildnew, subpreambles=true, width=0.8\textwidth]{no_fifo}\caption{If we instead set $\lw=2$, $\Uw=5$ and fix all other bounds to their trivial value, the instance is feasible, but no feasible schedule fulfills the FIFO property.}
    \label{fig: no fifo instance}
    \end{subfigure}
    \caption{Instances demonstrating that some worker might have to be on duty more often than others, and that the FIFO property is not compatible with global bounds.}
\end{figure} %
\section{NP-completeness}\label{sec: np completeness}

In this section, we prove that the general \dodosp is NP-complete. Furthermore, we identify the NP-complete special cases that arise from fixing certain bounds to their default values.

Before investigating the completeness, let us first check that the \dodosp is actually in NP. Naively encoding a feasible schedule as a spreadsheet takes $\Oequal{ND}$ bits, while any instance of the \dodosp can be encoded using $\Oless{D\log N}$ bits. Therefore we will have to find another certificate. Without any bounds on the requests, the \dodosp reduces to a flow problem. We will use such a flow as a certificate for the general problem.

\begin{lemma}\label{lem: dodosp without requests}
    The \dodosp restricted to instances with $\requestl{d}=0$ and $\requestu{d}=N$ for each day $d$ can be decided in polynomial time in $D$.
\end{lemma}
\begin{proof}
    We reduce this problem to a network flow problem in an acyclic graph $G=\lr{V, E}$ with
    \begin{align*}
        V=\lrb{s,t} &\cup\lrb{(1,\WORK,1,1),(1,\OFF,1,0)}\\
                    &\cup\lrb{\lr{d,\WORK, a, b}\,:\,2 \leq d \leq D,\,1\leq a\leq\uw,\,d-\Uo\leq b\leq\Uw}\\
                    &\cup\lrb{\lr{d,\OFF, a, b}\,:\,2\leq d \leq D,\,1\leq a\leq\uo,\,d-\Uo\leq b\leq\Uw}
    \end{align*}
    The work plan of a single worker can be represented as a sequence of these vertices of the form $s,\lr{1,\TASK_1,a_1,b_2},\dots,\lr{D,\TASK_D,a_D,b_D},t$. Besides $s$ and $t$, each vertex represents a possible state on a single day $d$ and contains three pieces of information:
    \begin{itemize}
        \item $\TASK_d$ encodes if the worker is on duty on day $d$ or not,
        \item day $d$ is the $a_d$-th day of the current $\TASK_d$ period and
        \item on the first $d$ days the worker was on duty $b_d$ times.
    \end{itemize}

    On the other hand such a sequence can be turned into a work plan of a single worker if and only if consecutive vertices $\lr{d,\TASK,a,b}$ and $\lr{d',\TASK',a',b'}$ are compatible with each other. This is the case if $d+1=d'$ and
    \begin{itemize}
        \item $\TASK=\TASK'=\WORK$ and $a+1=a'$ and $b+1=b'$ or
        \item $\TASK=\TASK'=\OFF$ and $a+1=a'$ and $b=b'$ or
        \item $\TASK=\WORK$, $\TASK'=\OFF$, $a\geq\lw$, $a'=1$ and $b'=b$ or
        \item $\TASK=\OFF$, $\TASK'=\WORK$, $a\geq\lo$, $a'=1$ and $b'=b+1$.
    \end{itemize}
    
    We add an edge from $\lr{d,\TASK,a,b}$ to $\lr{d',\TASK',a',b'}$ whenever they are compatible. We also add edges from $s$ to all tuples with $d=1$ and edges from all tuples with $d=D$ to $t$. Thus every feasible work plan of a single worker corresponds to an $s$-$t$-path in the constructed graph. %
    Note that this graph has no parallel edges and $\runtime{D^3}$ vertices and thus its size is polynomial in $D$. Furthermore, we reduce the size of the graph by imposing the conditions $a \leq d$ and $0 \leq b \leq d$.

    Feasible schedules correspond to (path decomposition of) integral $s$-$t$-flows of value $N$ in $G$, which can be computed in polynomial time. Hence it is easy to check if a feasible schedule exists.
\end{proof}

Now we can show that the \dodosp is in NP.

\begin{proposition}
    The \dodosp is in NP.
\end{proposition}
\begin{proof}
    As a certificate, we can use an integral $s$-$t$-flow of value $N$ in the same network as in the proof of \Cref{lem: dodosp without requests}. As already mentioned the network is of polynomial size.
    The number of workers that work on day $d$ equals the total throughput of all vertices $\lr{d,\WORK, a, b}$ with $a,b\in\CD$ and can therefore be efficiently checked for feasibility.
\end{proof}

To show NP-completeness of the \DODOSP, we restrict ourselves to instances with exact requests.
\begin{lemma}\label{prp: NP-hard}
    The \dodosp is strongly NP-complete, even if restricted to instances where $\Uo=\uw=\uo=D$, $\lo=1$ and $\request{d} \coloneqq \requestl{d} = \requestu{d} \in\lrb{0,1}$.
\end{lemma}
\begin{proof}
    For this proof, we will use a reduction from the strongly NP-complete \rthreepartition.

    \begin{problem}
        \problemtitle{\rthreepartition}
        \probleminput{A number $m\in \IN$ and a multiset  $A = \lrb{a_1,\dots,a_{3m}}\subset \IN$ of values with $\frac{1}{4}T < a_i < \frac{1}{2}T$ for all $a_i \in A$, where $T \coloneqq \frac{1}{m}\sum^{3m}_{i = 1}a_i$.\\[0.7cm]}
        \problemquestion{Does there exist a partition $A = \dot\bigcup_{i = 1}^m A_i$ with $|A_i| = 3$ and $\sum_{a\in A_i} a = T$ for all $1\leq i \leq m$?}
    \end{problem}

    Since the \rthreepartition is strongly NP-complete \parencite{gareyComputersIntractabilityGuide1979}, we may assume that $S = \sum_{i = 1}^{3m} a_i = mT \leq q\lr{m}$ for some universal polynomial $q$.\par\smallskip

    We construct an equivalent instance of the \dodosp. The idea of the construction is to identify the values $a_i$ with work periods of length $a_i$. Furthermore each subset $A_i$ of the partition will correspond to a worker who is on duty during the work periods that represent the elements of $A_i$.
    
    We set the number of days to $D = mT + 3m$ and the number of workers to $N = m$. Next, define $\request{d} = 0$ if $d=\sum_{i\in\lre{k}}\lr{a_i+1}$ for some $1\leq k\leq 3m$ and $\request{d}=1$ otherwise. Therefore, the requests are a unary encoding of the values $a_i$ with a $0$-request day serving as a separator between two consecutive values. Finally, we set $\Uw=T$, $\lw=\lceil\frac{1}{4}T\rceil$, and of course $U_o = \uw=\uo=D$ and $\lo=1$. An example of such a reduction is given in \Cref{fig: reduce 3-partition to dodosp}.
    \begin{figure}[htbp]
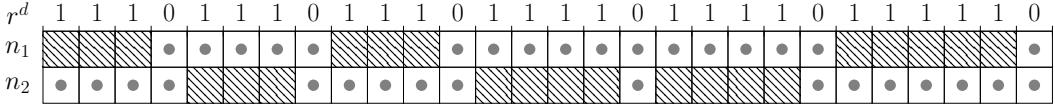

        \centering
        \includestandalone[width = \textwidth]{np_hardness_3_partition_2}
        \caption{Given a \threepartition instance with $m = 2$ and $A = \lrb{3,3,3,4,4,5}$, consider the \dodosp instance with the (exact) request shown in the figure as well as $\Uw=T=11$ and $\lw=\lceil\frac{1}{4}T\rceil=3$. Any solution to this \dodosp instance corresponds to a solution of the original problem. The shown solution correspond to the partition $\lrb{3,3,5},\lrb{3,4,4}$.}
        \label{fig: reduce 3-partition to dodosp}
    \end{figure}

    Any solution to the partition instance naturally translates to a feasible schedule. Indeed, given such a partition, we identify worker $n_i$ with the set $A_i$. The elements of $A_i$ correspond to three work periods. During these work periods $n_i$ is on duty while having all other days off.
    Since we are given a feasible partition, the upper bound $\Uw$ is respected. Furthermore, the lower bound $\lw$ is respected given that we consider the restricted version of the \threepartition and therefore $a_i\geq \lceil\frac{1}{4}T\rceil = \lw$ for all $i$.\smallskip

    To prove that the converse holds as well, namely that each feasible schedule can be transformed into a feasible partition, we assume that we are given a solution to the \dodosp instance and claim the following:
    \begin{claim}
        Every interval of days with request 1 representing a value $a_i$ will be covered by a single worker. 
    \end{claim}
    \begin{claimproof}
        Assume the claim would not hold, i.e., there exists an interval representing a value $a_i$ during which at least two workers $n_{j_1}$ and $n_{j_2}$ work.  Denote by $d_1 = \sum_{j<i}\lr{a_j+1}$ and $d_2 = \sum_{j\leq i}\lr{a_j+1}$ the day immediately before and immediately after the encoding of $a_i$. Both $n_{j_1}$ and $n_{j_2}$ do not work on those two days since the request for those days is 0. Hence, both of them work at least $\lw$ days in between $d_1$ and $d_2$. This is a contradiction since the total request in this time period is $a_i < \frac{1}{2}T \leq 2\cdot \lw$.
    \end{claimproof}
		From the bounds of the \dodosp, we can only deduce that the number of days a single worker works is \emph{at most} $\Uw=T$, but since the total request is $\sum^{3m}_{i = 1}a_i = mT = N\Uw$, this bound has to be tight for all workers in a feasible schedule. Therefore, a feasible schedule induces a partition of $A$ into sets $A_i$ with sum $T$ each. Finally, we have to prove that each $A_i$ has exactly three elements. This is a direct consequence of $\frac{1}{4}T<a_i<\frac{1}{2}T$.
		
    The size of the constructed instance and all occurring numbers are polynomially bounded in $m$ since $S \leq q(m)$ is bounded. Therefore, the reduction is polynomial.
\end{proof}

\begin{table}[htbp]
    \centering
    \renewcommand{\arraystretch}{1.4}
    \begin{tabular}{cll|c|l}
        req. bounds & \multicolumn{2}{c|}{additional bounds} & complexity & \multicolumn{1}{c}{proof}\\\hline
        $\requestu{d},\requestl{d}$&$U_x,l_y$ with $x,y\in \lrb{w,o}$ & or any superset & NP-compl. & \Cref{thm: NP-hard in general}\\
        $\requestu{d},\requestl{d}$&$\Uw,\Uo,\uw,\uo$ & or any subset & P & \Cref{thm: upper bounded dodosp}\\
        $\requestu{d},\requestl{d}$&$\uw,\uo,\lw,\lo$ & or any subset & P & \Cref{thm: local dodosp}\\[15pt]

        $\requestu{d}$&$\Uw,\uw,\lw,\lo$ & or any subset & P & \\ 
        $\requestu{d}$&$\Uw,\uo,\lw$ & or any superset & NP-compl. & \\
        $\requestu{d}$&$\Uo,l_x$ with $x\in \{o,w\}$ & or any superset & NP-compl. & \\
        $\requestu{d}$&$\Uw,\uo,\lo$ & or any superset & NP-compl. & 
    \end{tabular}
    \caption{Complexity of the \dodosp. Only the non-trivial request bounds are given. A lower bound on the requests is an upper bound on the number of workers who do not work on a given day. Therefore, we receive the same results for $\requestl{d}$, but with $w$ and $o$ swapped, for the last four rows.}%
    \label{tab: NP-hardness cases if request bounded by one side}
\end{table}

\begin{theorem}\label{thm: NP-hard in general}
    The \dodosp is NP-complete, even if $\requestl{d} = \requestu{d}$ for all $d\in \CD$ and if there is only one nontrivial global bound and one nontrivial local bound. 
\end{theorem}
\begin{proof}
    We have already proven this in the case that $\Uw$ and $\lw$ are nontrivial. The proofs for the special cases in which $\lo$ or $U_o$ are nontrivial work similarly, we just have to adapt the requests and the bounds: Fix an instance $\mathcal{I}$ of the \rthreepartition.\par
    
    If $\Uw$ and $\lo$ are nontrivial, we represent the values $a_i$ of $\mathcal{I}$ by off-periods. Thus the separator days that had a request of $0$ in the proof of \Cref{prp: NP-hard} will now have a request of $N$, while the request on all other days is $N-1$.
    Again, we set $\lo = \lceil\frac{1}{4}T\rceil$. Since for each worker the number of days off should sum up to $T$, we set $\Uw = D-T$. Here we use the fact that in the proof of \Cref{prp: NP-hard} all global bounds are fulfilled with inequality. 

    The case in which $U_o$ and $\lo$ are nontrivial is symmetric to the case in which $\Uw$ and $\lw$ are nontrivial. Hence, we again use requests of $N$ and $N-1$, set $U_o = T$, and $\lo = \lceil\frac{1}{4}T\rceil$.

    Finally, in the case that $U_o$ and $\lw$ are nontrivial, we again represent the values $a_i$ of $\mathcal{I}$ by work periods. We therefore use requests of $0$ and $1$ to build an equivalent instance of the \dodosp. Furthermore, we set $U_o = D-T$ and $\lw = \lceil\frac{1}{4}T\rceil$. 
\end{proof}

This proof can be adapted to other special cases whenever we have some proper encoding of the values $a_i$. This does not have to be a simple unary encoding. In fact using a less trivial encoding one can show similar complexity results for the case that the daily requests are only bounded from below or from above. The results are summarized in \Cref{tab: NP-hardness cases if request bounded by one side}.

\section{Upper bounds} \label{sec: upper bounds}

The primary goal of this section is to prove the following theorem:

\begin{restatable}{theorem}{thmudodosp}\label{thm: upper bounded dodosp}
The \UDODOSP can be decided in $\CO(D^2)$ time.
\end{restatable}
Because the \udodosp can be solved by an easier and faster algorithm if restricted to instances with exact requests, we first consider this case in \Cref{sec: upper bounds explicit requests}. In \Cref{sec: upper bounds request intervals}, we then adapt the arguments to the general \udodosp.

\subsection{Exact requests}\label{sec: upper bounds explicit requests}

In this subsection, we assume that $\request{d}\coloneqq\requestl{d} = \requestu{d}$. We exploit that all upper bounds \enquote{prefer} that workers swap between work and off periods as often as possible to evenly distribute the work days among the workers. If we ensure such an even distribution, it is sufficient to consider the total workload of certain consecutive days. 
We therefore define $\Request{d} = \sum_{i = 1}^d \request{i}$ to be the partially summed requests for $0 \leq d \leq D$.\par\medskip

There exist some easy-to-check inequalities which must be fulfilled by feasible instances.

\begin{lemma}\label{lem: ineq feas period counters upper bounds}
    Any feasible instance of the \udodosp with exact requests must fulfill the following four inequalities:
    \begin{align*}
        \Request{D} &\leq N\cdot \Uw\tag{req. $\Uw$}\label{ieq: RD <= NUw}\\
        N\cdot D - \Request{D} &\leq N\cdot \Uo \tag{req. $\Uo$}\label{ieq: ND-RD <= NUo}\\
        \Request{d+\uw}- \Request{d-1} &\leq N\cdot \uw &\forall 1 \leq d \leq D- \uw \tag{req. $\uw$}\label{ieq: <= Nuw}\\
        \Request{d+\uo} - \Request{d-1}&\geq  N &\forall 1 \leq d \leq D - \uo \tag{req. $\uo$}\label{ieq: >= N}
    \end{align*}
\end{lemma}
\begin{proof}
    Given a feasible instance of the \udodosp, we fix some feasible schedule. In this schedule each worker is at most $\Uw$ times on duty. Since feasibility implies that the total demand equals the total number of work days, summing over all workers gives \Cref{ieq: RD <= NUw}.

    Furthermore each worker has at most $\Uo$ days off, in other words they are at least $D-\Uo$ times on duty. Again summing over all workers gives \Cref{ieq: ND-RD <= NUo}.

    Fix any $1\leq d\leq D-\uw$. In the feasible schedule, no worker works on all days in the interval $d,\dots,d+\uw$, meaning each worker is on duty at most $\uw$ times. Thus, the total request of this interval is at most $N\cdot\uw$, proving \Cref{ieq: <= Nuw}.

    Fix any $1\leq d\leq D-\uo$. In the feasible schedule no worker can take the entire interval $d,\dots,d+\uo$ off, meaning each worker is on duty at least once during that interval. Thus, the total request of this interval is at least $N$, proving \Cref{ieq: >= N}.
\end{proof}

These inequalities are not only necessary but sufficient to characterize feasible instances of the \udodosp with exact requests. We will show this by stating \Cref{alg: dodosp with only upper bounds single number} which, given an instance fulfilling the inequalities in \Cref{lem: ineq feas period counters upper bounds}, returns a corresponding feasible schedule. 

In the algorithm, we expand the set of workers $\lrb{n_j\,:\, j\in\CN}$ to a set of worker-representatives $\lrb{n_j\,:\, j\in\N}$ and identify each worker $n_i$ with the representatives $n_{i+k\cdot N}$ for $k\in\N$. Each representative is on duty at most once. Finally, a worker is on duty whenever one of its representatives is.
\Cref{fig: representatives algo} provides a visualization of how the algorithm operates.%

\begin{minipage}{0.95\linewidth}%
\medskip%
\centering%
\begin{algorithm}[H]
    \SetKwInOut{Input}{input}
    \SetKwInOut{Output}{output}
    \Input{\udodosp instance which fulfills the inequalities in \Cref{lem: ineq feas period counters upper bounds}}
    \Output{A feasible schedule $\CN\times\CD\to\lrb{\text{ON},\text{OFF}}$.}
    \DontPrintSemicolon%
    \caption{\udodosp}\label{alg: dodosp with only upper bounds single number}
    \setstretch{1.35}
	Identify $n_j=n_{j+N}$\;
    \For{day $d = 1,\dots,D$}{
        assign exactly the workers with representatives $n_j$ for $j = \Request{d-1}+1,\dots, \Request{d-1} + \request{d}$ to work on day $d$\;
    }
	\Return computed schedule\;%
\end{algorithm}%
\medskip%
\end{minipage} 

The idea behind the algorithm is to assign the work days in a cyclic and the most evenly distributed way possible. This is also reflected in the two lemmas below.

Note that two representatives of the same worker cannot be assigned to work on the same day, since the request of each day is at most $N$.

\begin{lemma}\label{lem: two consec repr work}
    Assume that a worker $n_j$ is on duty on days $d_1$ and $d_2$, with $d_1 < d_2$. Then all other workers have to work on at least one of the days $d_1,\dots,d_2$. Furthermore, if the representatives $n_{j+kN}$ and $n_{j+\lr{k+1}N}$ are working on days $d_1$ and $d_2$, every other worker has a representative $n_i$ with $j+kN<i<j+\lr{k+1}N$ who works on one of these days.
\end{lemma}
\begin{proof}
    The assigned representatives form a consecutive sequence throughout the algorithm. We may assume that the representatives of $n_j$ that work on days $d_1$ and $d_2$ are consecutive. Thus they are the representatives $n_{j + kN}$ and $n_{j + (k+1)N}$ for some $k\in\N$. Due to the consecutiveness, the representatives $n_{j + kN + l}$ with $1\leq l < N$ are on duty on a day from $d_1$ to $d_2$, inclusive. Consequently, all other workers must also work at least one day.
\end{proof}

An informal rephrasing of this lemma would be: before a worker is an duty another time, all other workers must have been on duty at least once more.

\begin{figure}[htbp]
    \centering
    \includestandalone[width = 0.9\textwidth]{representatives_example}
    \caption{Illustration of \Cref{alg: dodosp with only upper bounds single number}. The assignment of work days to representatives is shown on the left. The corresponding schedule is shown on the right.}
    \label{fig: representatives algo}
\end{figure}

\begin{lemma}\label{lem: diff num shifts upp bound}
    The number of times two workers were on duty differs by at most one at any point throughout the algorithm. The same holds for the number of days off that two workers have had.  
\end{lemma}
\begin{proof}
    It is enough to show that if a worker $n$ has worked $w$ days at some point throughout the algorithm, then all the other workers have worked at least $w-1$ many days. This is true from the beginning up to and including the first day $n$ works. Now assume it is true after we assign $n$ to work on some day $d_1$ for the $w$'th time, meaning all other workers have been on duty at least $w-1$ times. Due to \Cref{lem: two consec repr work}, by the time $n$ works another time, all other workers must have also been on duty once more, thus, the inequality still holds.
\end{proof}

\begin{theorem}
    The schedule computed by \Cref{alg: dodosp with only upper bounds single number} satisfies the work requests and respects upper bounds $\uw$, $\uo$, $\Uw$ and $\Uo$.
\end{theorem}
\begin{proof}
    The work requests are satisfied by construction. Now assume the schedule that is computed by \Cref{alg: dodosp with only upper bounds single number} is not feasible and therefore violates a bound. We will show that this leads to a contradiction concerning the assumption that the instance fulfills the inequalities in \Cref{lem: ineq feas period counters upper bounds}. 

    \begin{enumerate}
        \item[$\Uw$:] Assume the computed schedule violates $\Uw$. Denote by $w_n$ the number of days worked by worker $n$. Since $\Uw$ is violated, there exists a worker $n'$ for which $w_{n'} > \Uw$. Given \Cref{lem: diff num shifts upp bound}, this implies that $w_{n} \geq \Uw$ for all workers $n \in \CN$. Since the computed schedule fulfills all requests exactly, it holds that
        \begin{equation*}
            \sum^{D}_{d = 1} \request{d} = \sum_{n \in \CN} w_n > N\cdot \Uw.
        \end{equation*}
        This contradicts \Cref{ieq: RD <= NUw} from \Cref{lem: ineq feas period counters upper bounds}. 

        \item[$\uw$:] Assume the computed schedule violates $\uw$. Hence, there exists a worker $n'$ and an interval $\mathcal{D}' = \lrb{d_1,\dots,d_{\uw+1}}$ of $\uw+1$ consecutive days on which worker $n'$ works.

        \Cref{lem: two consec repr work} applied to all consecutive pairs of days within $\mathcal{D}'$ implies that every other worker is on duty at least $\uw$ times in the interval $\mathcal{D}'$ (note that each application of \Cref{lem: two consec repr work} gives us a different representative of each worker).
        The situation is outlined in \Cref{fig: violation of u_w}.
        Hence, the summed request of days $d_1,\dots,d_{\uw+1}$ is
        \begin{equation*}
            \sum_{d = d_1}^{d_{\uw+1}} \request{d} > N\cdot \uw
        \end{equation*}
        which contradicts \Cref{ieq: <= Nuw} from \Cref{lem: ineq feas period counters upper bounds}. 

        \begin{figure}[htpb]
            \centering
            \includestandalone[width=0.5\linewidth]{uw_violation}
            \caption{Illustration of the case where $n_2$ violates the bound $u_w = 6$ with $N = 4$.}
            \label{fig: violation of u_w}
        \end{figure}
    \end{enumerate}
    With similar arguments one can show that the schedule computed by \Cref{alg: dodosp with only upper bounds single number} respects the bounds $\Uo$ and $\uo$ as well. Therefore, the algorithm returns a feasible schedule.
\end{proof}

 As observed before, this result implies that we can characterize instances with a feasible solution and determine feasibility in linear time. 

\begin{corollary}\label{cor: exact udodosp}
    An instance of the \udodosp has a feasible solution if and only if it fulfills the inequalities in \Cref{lem: ineq feas period counters upper bounds}. Thus, feasibility can be checked in $\CO(D)$. 
\end{corollary}

\begin{remark}
    Note that the running time of \Cref{alg: dodosp with only upper bounds single number} depends on the chosen output format. If an $N\times D$-matrix is returned, the running time is not polynomial but pseudopolynomial since $N$ is super-polynomial in the input size.  However, as already shown in \Cref{sec: np completeness}, polynomial representations for a schedule exist. 
    In this case, we can choose the cyclic interval $\left(\Request{d-1},\Request{d}\right]\,\text{mod}\,N$ for each day as an encoding of the workers on duty. Hence, if we want to determine if a worker $n$ works on day $d$, we must check if $n \in \left(\Request{d-1},\Request{d}\right]\,\text{mod}\,N$, which is equivalent to $1 \leq (n-R^{d-1} \mod N)\leq r^d$. 
    
    Note that this encoding does not work for the general \dodosp. More precisely, a feasible instance of the \dodosp may become infeasible if we restrict ourselves to solutions in which the workers on duty form a cyclic interval for each day.
\end{remark}

\subsection{Request intervals} \label{sec: upper bounds request intervals}
The results from \Cref{sec: upper bounds explicit requests} can be generalized to the case where we have request intervals for each day.

The goal is to compute the exact number $\work{d}$ of workers that should work on a given day, namely the value which will take the place of $\request{d}$. To this end, we demand that the values $\work{d}$ fulfill the inequalities in \Cref{lem: ineq feas period counters upper bounds}, together with $\requestl{d}\leq\work{d}\leq\requestu{d}$. Similar to before we again work with partial sums. Therefore, for $0\leq d\leq D$, we define
\begin{equation*}
    \Work{d} = \sum_{i=1}^d\work{i}.
\end{equation*}
This results in the following consequence of \Cref{cor: exact udodosp}:

\begin{corollary}\label{cor: udodosp feasibility}
    An instance of the \udodosp is feasible if and only if there are integers $\Work{d}$ for $0\leq d\leq D$ where
    \begin{align*}
        \Work{D}-\Work{0} &\leq N\cdot \Uw\\[5pt]
        ND + \Work{0} - \Work{D} &\leq N\cdot \Uo\\
        \Work{d+\uw}- \Work{d-1} &\leq N\cdot \uw &&\forall 1 \leq d \leq D- \uw \\
        N&\leq \Work{d+\uo} - \Work{d-1}&&\forall 1 \leq d \leq D - \uo\\
        \requestl{d}\leq\Work{d}-\Work{d-1}&\leq\requestu{d}&&\forall 1 \leq d \leq D.
    \end{align*}
\end{corollary}

Note that any solution can be shifted such that $\Work{0}=0$. Since each inequality bounds the difference of two variables by a constant, one can solve the inequality system described in \Cref{cor: udodosp feasibility} using any algorithm that computes a feasible potential in a digraph. Using the Moore-Bellman-Ford algorithm \parencite{bellmanRoutingProblem1958,fordNetworkFlowTheory1956,mooreShortestPathMaze1959} this can be done in $\CO(D^2)$ time.
\thmudodosp*

\begin{remark}
    The theorem only considers the decision problem, but we can also compute a feasible schedule if one exists. We first compute $\Work{d}$ for $0\leq d\leq D$ with $\Work{0}=0$. From then on, we continue as in \Cref{sec: upper bounds explicit requests}. Again the time complexity depends on the representation of the schedule.
\end{remark} %
\section{Local bounds}\label{sec: local bounds}
In this section, we consider the \ldodosp in where there are no bounds on the total number of days a single worker is on duty. Since to our knowledge, solving this problem is not any easier for instances with exact requests compared to general instances, we immediately deal with the general case. 
The main goal of this section is to prove the following:

\begin{restatable}{theorem}{thmldodosp}\label{thm: local dodosp}
The \LDODOSP can be decided in $\CO(D^2)$ time.
\end{restatable}

The key idea behind the proof is to consider the work periods and off periods, from a perspective similar to the one used in \parencite{thompsonImprovedImplicitOptimalModeling1995}. Instead of directly computing a schedule, we first determine how many work periods start and end on a given day. More precisely we introduce variables $\Startw{d}$ and $\Termw{d}$ for all $d\in\CD$. 
By $\Startw{d}$ we denote the number of work periods that start on some day $d'$, with $d'\leq d$.
By $\Termw{d}$, we denote the number of working periods that terminate \emph{before} day $d$, i.e., we count the work periods whose last day is at the latest $d-1$.

\begin{definition}\label{def: period counter}
    Given $D$ days and $N$ workers, integers $\lr{\Startw{d},\Termw{d}}_{d\in\CD}$ are called \emph{period counters} if there exists some work schedule for $N$ workers on $D$ days such that for each $d\in\CD$
    \begin{itemize}
        \item on the first $d$ days $\Startw{d}$ work periods begin and
        \item on the first $d-1$ days $\Termw{d}$ work periods end.
    \end{itemize}
    We say that the period counters \emph{represent} such a schedule.
    Note that, in general, this schedule can be any map $\CN\times\CD\to\lrb{\WORK, \OFF}$ and does not have to fulfill any bounds or requests. An example is depicted in \Cref{fig: period counter local bounds}.
\end{definition}

\begin{figure}[htbp]
    \centering
    \includestandalone[width=0.7\linewidth]{period_counters}
    \caption{A schedule together with the set of period counters which represent that schedule. Note that $\Termw{d}$ denotes the work periods which end on the first $d-1$ days. The corresponding off period counters are given in gray.}
    \label{fig: period counter local bounds}
\end{figure}

One could also use $\Starto{d}$ and $\Termo{d}$ to count the off periods instead. Even though it is not necessary to count both, it is sometimes convenient to do so. Therefore, we begin by pointing out the connection between $\lr{\Startw{d},\Termw{d}}_{d\in\CD}$ and $\lr{\Starto{d},\Termo{d}}_{d\in\CD}$.

\begin{lemma}\label{lem: off <-> work}
    Given (work) period counters $\lr{\Startw{d},\Termw{d}}_{d\in\CD}$ and corresponding off period counters $\lr{\Starto{d},\Termo{d}}_{d\in\CD}$ the following equations hold:
    \begin{align*}
        \Startw{1}+\Starto{1} &= N&&\\
        \Starto{d}&=\Termw{d}+\Starto{1}=\Termw{d}+N-\Startw{1}&&\forall 1\leq d\leq D\\
        \Termo{d}&=\Startw{d}-\Startw{1}&&\forall 1\leq d\leq D
    \end{align*}
\end{lemma}
\begin{proof}
    On the first day, each worker starts either a work period or an off period.
    
    From the second day onward a worker starts an off period if and only if that worker terminates a work period on the previous day.
    
    Similarly, from the second day onward, a worker starts a work period if and only if this worker terminates an off period on the previous day.
\end{proof}

Of course not every sequence of numbers forms period counters. For example, there cannot be more than $N$ periods starting on the same day. Some similar conditions are captured in \Cref{lem: period counter}.

\begin{lemma}\label{lem: period counter}
    Period counters $\lr{\Startw{d},\Termw{d}}_{d\in\CD}$ fulfill $\Termw{1}=0$ and the following inequalities for all $1\leq d\leq D-1$:
    \begin{align}
        \tag{non decr.}\label{ieq: monotonicity} \Startw{d}&\leq\Startw{d+1},\quad\Termw{d}\leq\Termw{d+1}\\
        \tag{pos. work}\label{ieq: start before termination}
        \Termw{d+1}&\leq\Startw{d}\\
        \tag{pos. off} \label{ieq: break before next}
        \Startw{d+1}&\leq\Termw{d}+N
    \end{align}
\end{lemma}
\begin{proof}
    No work period can terminate before day $1$, hence $\Termw{1}=0$.
    On each day a non-negative number of work periods start and terminate. As $\Startw{d}$ and $\Termw{d}$ are partial sums of these non-negative values, they must be non-decreasing and therefore \Cref{ieq: monotonicity} must hold.

    Each work period has a positive length. Since each work period counted by $\Termw{d+1}$ ends on day $d$ or earlier, this work period has to start on day $d$ or earlier and is therefore also counted by $\Startw{d}$. This proves \Cref{ieq: start before termination}.

    For the same reason, it is the case that $\Termo{d+1}\leq\Starto{d}$. Together with \Cref{lem: off <-> work} this gives us \Cref{ieq: break before next}.
\end{proof}

So far, we have considered period counters that might correspond to any assignment $\CN\times\CD\to\lrb{\WORK, \OFF}$. However, we are primarily interested in feasible assignments.

\begin{definition}\label{def: feasible period counter}
    Given an \ldodosp instance, period counters $\lr{\Startw{d},\Termw{d}}_{d\in\CD}$ are called \emph{feasible} if they represent a feasible schedule.
\end{definition}

\begin{lemma}\label{lem: feasible period counter}
    Feasible period counters $\lr{\Startw{d},\Termw{d}}_{d\in\CD}$ fulfill the following inequalities:
    \begin{align}
        \tag{bound. $\lw$}\label{ieq: lw boundary}
        \Termw{\lw}&=0,\quad\Startw{D-\lw+1}=\Startw{D}&&\\
        \tag{count. $\lw$}\label{ieq: lw}
        \Termw{d+\lw}&\leq\Startw{d}&&\forall 1\leq d\leq D-\lw\\
        \tag{count. $\uw$}\label{ieq: uw}
        \Startw{d}&\leq\Termw{d+\uw}&&\forall 1\leq d\leq D-\uw\\
        \tag{bound. $\lo$}\label{ieq: lo boundary}
        \Startw{1}&=\Startw{\lo},\quad\Termw{D-\lo+1}=\Termw{D}&&\\
        \tag{count. $\lo$}\label{ieq: lo}
        \Startw{d+\lo}-N&\leq\Termw{d}&&\forall 1\leq d\leq D-\lo\\
        \tag{count. $\uo$}\label{ieq: uo}
        \Termw{d}+N&\leq\Startw{d+\uo}&&\forall 1\leq d\leq D-\uo\\
        \tag{count. req.}\label{ieq: request}
        \requestl{d}&\leq\Startw{d}-\Termw{d}\leq\requestu{d}&&\forall 1\leq d\leq D
    \end{align}
\end{lemma}
\begin{proof}
    If a work period would terminate before day $\lw$ or start after day $D-\lw+1$ then it would be shorter than $\lw$ and thus \Cref{ieq: lw boundary} must hold.
    Only the work periods that start on day $d$ or earlier are allowed to terminate before day $d+\lw$, which gives \Cref{ieq: lw}.
    Similarly, each work period that starts at day $d$ or earlier must terminate before day $d+\uw$, which immediately gives \Cref{ieq: uw}.
    With the same arguments we obtain \Cref{ieq: lo boundary} as well as
    \begin{equation*}
        \Termo{d+\lo}\leq\Starto{d}\quad\text{and}\quad \Starto{d}\leq\Termo{d+\uo}.
    \end{equation*}
    Using \Cref{lem: off <-> work} this gives us \Cref{ieq: lo} and \Cref{ieq: uo}.
    
    For \Cref{ieq: request} it is sufficient to observe that the number of workers that are on duty on day $d$ equals the number of work periods that contain day $d$. These are exactly the work periods that start on day $d$ or earlier but do not terminate before day $d$.
\end{proof}

Next, we show that the inequalities stated thus far are not only necessary but also sufficient for period counters to be feasible.
To prove this, we restrict ourselves to schedules with the following FIFO property: whenever a work period $W_1$ starts earlier than another work period $W_2$, $W_1$ does not terminate later than $W_2$. The same should hold for off periods.

\begin{remark}
    Note that in the general \dodosp, we cannot restrict ourselves to such schedules. For example, the instance shown in \Cref{fig: reduce 3-partition to dodosp} has no feasible schedule in which the off periods fulfill the described property. 
\end{remark}

To prove that the inequalities in \Cref{lem: feasible period counter} are not only necessary, but also sufficient for feasibility, we explicitly construct a schedule from period counters that satisfy the inequalities and show that this schedule is feasible.
To this end, we again expand the set of workers $\lrb{n_j\,:\, j\in\CN}$ to a set of worker-representatives $\lrb{n_j\,:\, j\in\N}$ and identify each worker $n_i$ with the representatives $n_{i+k\cdot N}$ for $k\in\N$. Each representative will obtain at most one work period. Finally, as before, a worker is on duty whenever one of its representatives is. The construction is given in \Cref{alg: feasible period counters to schedule}. For the period counters in \Cref{fig: period counter local bounds}, \Cref{alg: feasible period counters to schedule} would compute exactly the depicted schedule.

\begin{minipage}{0.95\linewidth}%
\medskip%
\centering%
\begin{algorithm}[H]
    \SetKwInOut{Input}{input}
    \SetKwInOut{Output}{output}
    \Input{\ldodosp instance, period counters which fulfill the inequalities in \Cref{lem: period counter} and \Cref{lem: feasible period counter}}
    \Output{A feasible schedule $\CN\times\CD\to\lrb{\text{ON},\text{OFF}}$}
    \DontPrintSemicolon%
    \caption{\ldodosp period counters to schedule}\label{alg: feasible period counters to schedule}
    \setstretch{1.35}
	Identify $n_j=n_{j+N}$\;
    \For{day $d = 1,\dots,D$}{
        assign worker/representative $n_j$ to work if and only if $\Termw{d}<j\leq\Startw{d}$\;
    }
	\Return computed schedule\;%
\end{algorithm}
\medskip%
\end{minipage}

\begin{lemma}\label{lem: one period per representative}
    The output of \Cref{alg: feasible period counters to schedule} has the following properties:
    \begin{enumerate}
        \item\label{itm: relevant representatives} Only the representatives $n_1$ to $n_{\Startw{D}}$ are assigned to work.
        \item\label{itm: each representative one period} Each of the representatives $n_1$ to $n_{\Startw{D}}$ is assigned exactly one work period.
        \item\label{itm: period representative correspondence} The $i$-th work period of a worker equals the work period associated with their $i$-th representative.
    \end{enumerate}
\end{lemma}
\begin{proof}
    \Cref{itm: relevant representatives} follows from $\Termw{0}=0$ and \Cref{ieq: monotonicity}.
    
    Assume for some $1\leq j\leq\Startw{D}$ that the representative $n_j$ is never assigned to work. Then for all $d\in \CD$, either $j \leq \Termw{d}$ or $j > \Startw{d}$ hold. Let $d'$ be the largest $d$ with $j > \Startw{d}$. This is well defined, since $\Startw{d}$ is non-decreasing in $d$. Because $j \leq \Startw{D}$, it follows that $d' < D$. Then $j \leq \Termw{d'+1}$ must hold, which contradicts \Cref{ieq: start before termination} and consequently gives \Cref{itm: each representative one period}.
    
    Again by \Cref{ieq: monotonicity}, the days a representative is assigned to work are consecutive. Let $n_{j_1}$ and $n_{j_2}$ with $1\leq j_1<j_2\leq\Startw{D}$ be representatives of the same worker. Let $d$ be the earliest day $d$ such that $j_2\leq\Startw{d+1}$. By the choice of $d$, the representative $n_{j_2}$ does not work before day $d+1$. Furthermore:
    $$
    j_1\leq j_2-N\leq\Startw{d+1}-N\overset{\Cref{ieq: break before next}}{\leq}\Termw{d}.
    $$
    Hence, $n_{j_1}$ cannot work on any day later than $d-1$. Therefore the work periods of $n_{j_1}$ and $n_{j_2}$ are not only disjoint, but they are even separated by at least one day off and in canonical order. Thus, each work period of a representative is a distinct work period of the corresponding worker, which implies \Cref{itm: period representative correspondence}.
\end{proof}

Note that the work period associated with representative $n_j$ starts with the first day $d$ that satisfies $j\leq\Startw{d}$ and ends on the last day $d'$ for which $\Termw{d'}<j$.

\begin{proposition}\label{prp: counters to schedule correctness}
    The work schedule computed by \Cref{alg: feasible period counters to schedule} satisfies the work requests and respects local bounds $\lw$, $\uw$, $\lo$, and $\uo$.
\end{proposition}
\begin{proof}
    The number of workers that work on day $d$ equals the number of representatives that work on day $d$. By construction, there are $\Startw{d}-\Termw{d}$ representatives on duty, which is a feasible value by \Cref{ieq: request}.
    Next we check the local bounds.
    \begin{itemize}
        \item[$\lw$]
            Due to \Cref{ieq: lw boundary,ieq: monotonicity} no work period terminates before day $\lw$ and no work period starts after day $D-\lw+1$. Thus work periods that start on day $1$ or finish on day $D$ cover at least $\lw$ many days.
        
            Fix $1\leq d\leq D-\lw$ and $j$ such that $n_j$ is a representative that starts working on day $d+1$, and therefore $\Startw{d}<j$ and $j\leq\Startw{d+1}$. This implies 
            $$
            \Termw{d+\lw}\overset{\Cref{ieq: lw}}{\leq}\Startw{d}<j\leq\Startw{d+1}\overset{\Cref{ieq: monotonicity}}{\leq}\Startw{d+\lw}.
            $$
            Thus the representative $n_j$ is on duty on day $d+\lw$ and the corresponding work period covers at least $\lw$ many days.
        \item[$\uw$]
            Given a representative $n_j$, let $d$ be the day $n_j$ starts working. This implies $j\leq\Startw{d}$.
            If $d>D-\uw$, then the work period cannot exceed $\uw$ days anyway.
            If $d\leq D-\uw$, we get
            $$
            j\leq\Startw{d}\overset{\Cref{ieq: uw}}{\leq}\Termw{d+\uw}.
            $$
            Thus the representative $n_j$ does not work on day $d+\uw$ and therefore the corresponding work period lasts at most $\uw$ days.
        \item[$\lo$]
            By \Cref{ieq: lo boundary}, no work period starts on days $2,\dots,\lo$ and no work period terminates on days $D-\lo+1,\dots,D-1$. Hence no off period terminates before day $\lo$ and no off period starts after day $D-\lo+1$. Thus, off periods that start on the first day or finish on day $D$ last for at least $\lo$ days.

            Fix $1\leq d\leq D-\lo$ and $j$ such that $n_j$ is a representative whose last work day is $d$ and thus $\Termw{d}<j$. This implies
            $$
            \Startw{d+\lo}\overset{\Cref{ieq: lo}}{\leq}\Termw{d}+N<j+N.
            $$
            Hence the next representative of the same worker does not yet work on day $d+\lo$. In other words: the off period in between lasts for at least $\lo$ days.
        \item[$\uo$]
            Using \Cref{ieq: uo} for $d=1$ we obtain that 
            \begin{equation*}
                N = \Termw{1} + N \leq \Startw{1+u_o}.
            \end{equation*}
            Hence, at least the first $N$ representatives and therefore all workers start their first work period no later than day $\uo+1$. Therefore, all off periods that start on day $1$ are sufficiently short. Let $n_j$ be a representative that finishes working on day $d-1$ and therefore $j\leq\Termw{d}$.
            If $d>D-\uo$, the following off period cannot exceed $\uo$ days anyway. If $d\leq D-\uo$, we get
            $$
            j+N\leq\Termw{d}+N\overset{\Cref{ieq: uo}}{\leq}\Startw{d+\uo}.
            $$
            Thus the next representative $n_{j+N}$ of the same worker starts working no later than day $d+\uo$ and therefore the off period in between lasts at most $\uo$ days.\qedhere
    \end{itemize}
\end{proof}

The correctness of \Cref{alg: feasible period counters to schedule} implies that the inequalities stated in \Cref{lem: period counter} and \Cref{lem: feasible period counter} are not only necessary, but also sufficient. This is captured by \Cref{cor: iff inequalities}.

\begin{corollary}\label{cor: iff inequalities}
    Period counters are feasible if and only if they satisfy the inequalities in \Cref{lem: period counter} and \Cref{lem: feasible period counter}.
\end{corollary}

Therefore, checking the feasibility of an \ldodosp instance is equivalent to finding an \emph{integral} solution to the system of inequalities given by \Cref{lem: period counter} and \Cref{lem: feasible period counter}.
Using the same techniques as in \Cref{sec: upper bounds request intervals}, we can solve this integral system of inequalities in polynomial time.

\begin{proposition}\label{prp: solve ieq system}
    The system of inequalities given in \Cref{lem: period counter} and \Cref{lem: feasible period counter} can be solved in $\runtime{D^2}$.
\end{proposition}

Now we can conclude with the proof of our main theorem. 

\thmldodosp*
\begin{proof}
    By \Cref{prp: solve ieq system} we can check in $\runtime{D^2}$ if the system of inequalities given by \Cref{lem: period counter} and \Cref{lem: feasible period counter} is feasible. Due to \Cref{cor: iff inequalities}, this also determines whether the \ldodosp instance is feasible.
\end{proof}

\begin{remark}
    Using \Cref{alg: feasible period counters to schedule} we can translate feasible period counters into a feasible schedule $\CN\times\CD\to\lrb{\text{ON},\text{OFF}}$. Once again, the running time depends on the output format we choose. One can, as before, return the cyclic interval $\left(\Termw{d},\Startw{d}\right]$ for each day $d$ instead of listing all workers that are on duty, which results in a polynomial algorithm.
\end{remark} %

\section{Optimizing the number of workers} \label{sec: optimization}

Until now, we have only considered the decision version of the \dodosp. However, in practice the number of workers might not be given, but instead should be determined by the algorithm.

Since an instance of the \dodosp can be infeasible due to the value of $N$ being either too large or too small, both minimizing and maximizing $N$ appear reasonable. We will only consider minimizing the number of workers, given that both can be achieved using almost the same techniques.
Further, because optimizing $N$ is as least as hard as checking feasibility, we restrict ourselves to the special cases of the \dodosp we can solve in polynomial time: the \udodosp and the \ldodosp.

For both special cases, we can classify feasible instances by checking if a negative cycle (or a feasible potential) in a weighted digraph exists. In the following, we refer to this graph as the \emph{potential graph}.

\begin{lemma}\label{lem: N is an interval}
Given an instance of the \udodosp or the \ldodosp, the set of (fractional) values of $N$ for which the potential graph contains no negative cycle is a possibly empty or one-sided interval.
\end{lemma}
\begin{proof}
    Since the sum of the weights of a cycle is linear in $N$ and should be non-negative, the set of feasible values for $N$ can be expressed as one-dimensional polyhedron with a linear constraint for each cycle in the potential graph.
\end{proof}

One can minimize $N$ with a binary search. For this, we must first of all find an upper and a lower bound on the optimal value of $N$, if existant. Clearly 0 is a lower bound on $N$, hence it is left to find an upper bound.

\begin{lemma}
    The optimal value for $N$ is bounded from above by $\sum_{d\in\CD}\requestl{d}$. 
\end{lemma}
\begin{proof}
    Fix some feasible schedule for the minimum value of $N$. For each day $d$ mark $\requestl{d}$ many workers that are on duty on day $d$. There are at most $\sum_{d\in\CD}\requestl{d}$ many marked workers. If we remove all workers that are not marked, we still have a feasible schedule. Thus, each worker was marked, and therefore $N\leq \sum_{d\in\CD}\requestl{d}$.
\end{proof}

Furthermore, we must be able to decide if a given value of $N$ is (i) too small (ii) feasible (iii) too large or (iv) the instance is infeasible independent of $N$. Using the Moore-Bellman-Ford algorithm \parencite{bellmanRoutingProblem1958,fordNetworkFlowTheory1956,mooreShortestPathMaze1959} on the potential graph, we can distinguish between cases (i) to (iv): if the algorithm finds a feasible potential, we are in case (ii). If it returns a negative cycle, consider the affine function $a\cdot N+b$ that describes the total cost of this cycle in terms of $N$. Depending on the sign of $a$, we know whether $N$ is too large or too small. If $a=0$, the instance is infeasible for every value of $N$.

\begin{remark}
    Note that the binary search to optimize $N$ is essentially the ellipsoid method on the one-dimensional polyhedron described in \Cref{lem: N is an interval}. Distinguishing cases (i) to (iii) is basically a separation oracle for this polyhedron.
\end{remark} %
\section{Summary} \label{sec: summary}

In this paper, we have considered the \dodosp and examined its complexity. We have shown that if bounds on the requests from both sides are given, then the problem is NP-complete if and only if one global and one local lower bound is nontrivial. Furthermore, there are two special cases in which the \dodosp is polynomial-time solvable, as shown in \Cref{tab: NP-hardness cases if request bounded by one side}. The algorithms we introduced to solve these cases are easy to implement and fast. In case the requests are bounded only from one side, the complexity results only differ slightly.

Our results provide a straightforward method for distinguishing between computationally easy subcases and those that are more complex. This distinction is particularly valuable in real-world applications, where our results may help guide decisions concerning which constraints to relax or drop in order to make the problem computationally more manageable. 
Knowing the specific constraints that drive complexity is not only helpful to determine parameters to relax, but is a valuable information when designing a model in the first place. Even though the solution quality may decrease by choosing a weaker model, this is made up for by the speedup in running time, which is especially important if schedules must be computed quickly.

Even though our algorithms can only solve special cases of the \dodosp they can be of use for the general case, and even other variants of personell rostering. Since the algorithms are very efficient, one can use them as starting point or as a subroutine for more sophisticated algorithms. \medskip 

Future research could explore several interesting directions. One potential generzalization is the introduction of heterogeneous workers, where each worker's work pattern is subject to different constraints, such as varying limits on consecutive work days or total days off. This would make the problem more reflective of real-world scenarios, where employees often have individual preferences.
Another important extension would involve handling predetermined entries in the schedule, such as pre-scheduled holidays or leave requests for certain workers.
This might include continuing an existing schedule to a larger time period.
From the perspective of \Cref{sec: optimization}, one can ask if the \dodosp can be approximated regarding certain parameters such as the minimal number of required workers.
Finally, one could analyze real-world instances hoping to find a more restricted problem definition that covers practical instances, without being NP-hard. \clearpage
\section*{Acknowledgments}
We would like to express our gratitude to Erik Saule for highlighting that the problem is not trivially in NP, and to Maximilian von Aspern for proposing an NP-certificate. 
We also thank Sigrid Knust, Sebastian Lüderssen, Meike Neuwohner, Pieter Smet, Vera Traub, and Greet Vanden Berghe for their valuable comments on various aspects of the content and overall development of the paper.
Finally, we thank Luke Connolly (KU Leuven) for his editorial consultation.
This research was supported by KU Leuven C24E/23/012 `Human-centred decision support based on new theory for personnel rostering'.

\printbibliography
\clearpage
\appendix
\section{Appendix}

\subsection{NP-completeness proof in \texorpdfstring{\parencite{brunnerBoundedFlexibility2013}}{citeBrunner}} \label{sec: wrong proof correction}
In their proof \citeauthor{brunnerBoundedFlexibility2013} aim to reduce the strongly NP-complete \CP to the optimization version of the \dodosp. %
\begin{problem}
    \problemtitle{\CP (CP)}
    \probleminput{A circulant $m\times n$-matrix $A$ and a nonnegative $m$-dimensional vector $b$.\\[0.3cm]}
    \problemquestion{Find an $n$-dimensional nonnegative integer vector $x$ which minimizes $\mathbf{1}\cdot x$ with $Ax \geq b$.}
\end{problem}
Here, an $m\times n$-matrix $A$ is called \emph{circulant} if there exists an $m$-dimensional vector $a$ such that all columns of $A$ are only copies of $a$ whose entries were shifted in a circular way.

However, \citeauthor{brunnerBoundedFlexibility2013} ultimately use the wrong direction of reduction and try to solve certain instances of the \dodosp using the CP. More precisely, they consider instances of the optimization version of the \dodosp with $\uw = \lw$, $\uo = \lo$, $r_u^d = D$ for all days (meaning the requests are only bounded from below) and $D = c\left(u_w + u_o\right)$ for some non-negative integer $c$. Indeed, if one uses the cyclic definition of the optimization version of the \dodosp, which was not originally done by \citeauthor{brunnerBoundedFlexibility2013}, one can reduce such an instance $I$ to an instance $I'$ of CP in the following way. 

Let $a$ be the vector of length $c\lr{u_w+u_o}$ which consists of alternating sequences of $u_w$ ones and $u_o$ zeros. Furthermore, define $A$ to be the circulant matrix which arises from rotating $a$ by onre entry each time. In fact, $A$ will only have $u_w + u_o$ columns, since after this many rotations the vector $a$ is sent to itself. Finally, we set $b_d = r_l^d$. 
The columns in $A$ denote the only possible schedules for a worker in the given instance of the optimization version of the \dodosp. Therefore, any solution to $I'$ naturally provides a solution to $I$ and vice versa. 

However, since this reduction only works for a particular kind of instances of the (circular) optimization version of the \dodosp, we do not get a statement on the complexity of the CP depending on the complexity of the (circular) optimization version of the \dodosp.

\subsection{Missing proofs of \Cref{sec: np completeness}} \label{sec: proofs np completeness}

In this subsection we want to state the proofs for all the NP-completeness results that were only mentioned in \Cref{sec: np completeness}.

So far we only reduced the \threepartition to the \dodosp with exact requests which then implied NP-completeness for the general \dodosp. However, if either one of $\requestl{d}$ and $\requestu{d}$ is trivial, meaning the requests are bounded by only one side, then the problem is not NP-hard in exactly the same cases as the general \dodosp. 
We will determine the computational complexities for all special cases that arise from fixing certain bounds to their trivial values. Most of these can be derived from the general case, but we have to analyze some special cases. An overview of the complexities is again given in \Cref{tab: NP-hardness cases if request bounded by one side detailed}.\smallskip

\begin{table}[htbp]
    \centering
    \renewcommand{\arraystretch}{1.4}
    \begin{tabular}{cll|c|l}
        req. bounds & \multicolumn{2}{c|}{additional bounds} & complexity & \multicolumn{1}{c}{proof}\\\hline
        $\requestu{d},\requestl{d}$&$U_x,l_y$ with $x,y\in \lrb{w,o}$ & or any superset & NP-compl. & \Cref{thm: NP-hard in general}\\
        $\requestu{d},\requestl{d}$&$\Uw,\Uo,\uw,\uo$ & or any subset & P & \Cref{thm: upper bounded dodosp}\\
        $\requestu{d},\requestl{d}$&$\uw,\uo,\lw,\lo$ & or any subset & P & \Cref{thm: local dodosp}\\[15pt]

        $\requestu{d}$&$\Uw,\uw,\lw,\lo$ & or any subset & P & \Cref{thm: complexity one sided}\\
        $\requestu{d}$&$\Uw,\uo,\lw$ & or any superset & NP-compl. & \Cref{lem: NP-hard rud Uw uo lw}\\
        $\requestu{d}$&$\Uo,l_x$ with $x\in \{o,w\}$ & or any superset & NP-compl. & \Cref{lem: NP-hard Uo l}\\
        $\requestu{d}$&$\Uw,\uo,\lo$ & or any superset & NP-compl. & \Cref{lem: NP-complete Uw uo lo}
    \end{tabular}
    \caption{Complexity of the \dodosp. Only the non-trivial request bounds are given. A lower bound on the requests is an upper bound on the number of workers who do not work on a given day. Therefore, we receive the same results for $\requestl{d}$, but with $w$ and $o$ swapped, for the last four rows.}%
    \label{tab: NP-hardness cases if request bounded by one side detailed}
\end{table}

\begin{lemma}\label{lem: NP-hard Uo l}
    The \dodosp is NP-complete, even if only $\requestu{d}, \Uo$ and one of the local lower bounds $\lw$ or $\lo$ is given. 
\end{lemma}
\begin{proof}
    Consider the instances constructed in \Cref{thm: NP-hard in general}. We keep $\requestu{d}=\request{d}$, but set $\requestl{d}=0$ for all $d\in\CD$.
    By construction, we have 
    \begin{equation*}
        \sum_{d\in \CD}N-\requestu{d} = \Uo\cdot N.
    \end{equation*}
    Hence no day $d$ can have fewer than $\requestu{d}$ workers working. Thus, the set of feasible schedules did not change and the reduction from the \rthreepartition works.
\end{proof}

\begin{lemma}\label{lem: NP-hard rud Uw uo lw}
    The \dodosp is NP-complete even if only $\requestu{d}, \Uw, \uo$, and $\lw$ are given. 
\end{lemma}
\begin{proof}
    We follow the proof of \Cref{prp: NP-hard}, but use a slightly more complex encoding of the values $a_i$. Without loss of generality let $T$ be divisible by $4$. We set $\lw=\frac{T}{4}$ and $\uo=\frac{T}{2}$. This , in particular, implies $\lw<a_i<2\lw=\uo$ for all $i$.
    
    The encoding of $a_i$ consists of the following sequence: $\uo$ times zero, $\lw$ times $N$, $\lw$ times one, $\lw$ times zero, $\lw$ times $\lr{N-1}$, $a_i$ times one; after that we repeat the sequence in reversed order by adding $\lw$ times $\lr{N-1}$, $\lw$ times zero, $\lw$ times one, $\lw$ times $N$ and finally $\uo$ times zero (see \Cref{fig: NP-hard special case sequence}).

    \begin{figure}[htbp]
        \centering
        \includestandalone{np\_hardness\_sequence\_special\_case}
        \caption{Sequence of values $\requestu{}$ to encode $a_i$ (up to a mirrored part to the right). The rectangles represent the first to fourth block of that sequence.}
        \label{fig: NP-hard special case sequence}
    \end{figure}

    We claim the following: 
    \begin{claim}
        In any feasible schedule for such a sequence, there is one \enquote{special} worker who works on $4\lw + a_i$ many days and all other workers work on $4\lw$ many days. There always exists a schedule that fulfills the local bounds.
    \end{claim}

    We can put the encodings of the values $a_i$ one after another and merge trailing and leading zeros of consecutive encodings. This way, the encoding sequences are still independent and do not restrict each other. According to the claim, we define $U_w = 3m\cdot4\lw + T$ and we use the same arguments as in the general proof. 

    It remains to prove the claim.
    \begin{claimproof}
        We call the subsequence from position $\uo+1$ to position $\uo + \lw$ the first block of that sequence. This subsequence contains exactly the $\lw$ positions with request $N$ (in the left half of the sequence). Similarly we denote by the second block the following $\lw$ positions with request one, by the third block the following $\lw$ positions with request zero and by the fourth block the $\lw$ positions with request $N-1$ (cf. \Cref{fig: NP-hard special case sequence}). The sequence of $a_i$ ones will be abbreviated to $a_i$.
        
        No worker works in the third block.
        Due to the leading zeros and $\uo$, every worker is on duty on the first day of the first block. Therefore, by $\lw$ every worker is on duty for the entire first block.
        Whichever worker works on any day of the second resp. the fourth block has to work on the first day of the second resp. the last day of the fourth block ($\lw$).
        Hence at most one worker works in the second block and at most $N-1$ workers work in the fourth block.
        Again considering $\uo$, we conclude that exactly one worker $n_1^*$ is on duty for the entire second block while all other $N-1$ workers are on duty for the entire fourth block. Additionally, $n_1^*$ has to work the first $\lw$ days of $a_i$.
        The same holds for the four blocks after $a_i$. Thus there is also a special worker $n_2^*$ who works in the second to last block (the block with $\requestu{}=1$) and works the last $\lw$ days of $a_i$. Since $a_i<2\lw$, both $n_1^*$ and $n_2^*$ are on duty in the middle of $a_i$, hence they are the same worker.
        Thus we have shown that every feasible schedule has one special worker who works for $4\lw+a_i$ days (first and second block, $a_i$, second to last and last block)
        while all other workers work for $4\lw$ days (first, fourth, fourth to last and last block).
        Finally, note that this indeed gives a feasible schedule.
    \end{claimproof}
    As already mentioned this ends the proof of the lemma.
\end{proof}

\begin{lemma}\label{lem: NP-complete Uw uo lo}
    The \dodosp is NP-complete, even if only $\requestu{d}, \Uw$, $u_o$ and $l_o$ are given.
\end{lemma}
\begin{proof}
    We again follow the proof of \Cref{thm: NP-hard in general} and assume w.l.o.g.\;$T$ to be even, otherwise we can multiply all values $a_i$ and thus $T$ with two. This time we set $\lo=T$ and $\uo=\frac{3}{2}T+1$. The encoding of $a_i$ is given by \Cref{fig: NP-hard special case 2 sequence}.
    \begin{figure}[htbp]
        \centering
        \includestandalone{np\_hardness\_sequence\_2}
        \caption{Sequence of values $\requestu{}$ to encode $a_i$ (up to a mirrored part to the right)}
        \label{fig: NP-hard special case 2 sequence}
    \end{figure}
    It suffices to prove the following:
    \begin{claim}
        In any feasible schedule for such a sequence, there is one \enquote{special} worker who works on $4 + a_i$ many days and all other workers work on $4$ days.  There always exists a schedule that fulfills the local bounds.
    \end{claim}
    \begin{claimproof}
        Due to the leading zeros, everybody works on the day where $\requestu{}=N$.
        Let $d_1$ and $d_2$ be the first two days where $\requestu{}=1$ (meaning $d_2$ is the first day of $a_i$).
        There are $\uo$ days strictly in between $d_1$ and $d_2$ whose $\requestu{}$ values sum up to $N-1$. Thus there has to be a special worker $n_1^*$ who works on $d_1$. Furthermore, all other workers have to work on the day where $\requestu{}=N-1$, and $n_1^*$ is on duty on $d_2$. Since $\frac{1}{2}T+a_i<T=\lo$, none of the other workers can work on any of the $a_i$-days.

        The same situation holds for the mirrored sequence to the right of $a_i$ with a special worker $n_2^*$. Since none of the other workers may work at the end of $a_i$, we again know that $n_1^*=n_2^*$. Thus there is one special worker who works on the first and the last day of $a_i$. Since $a_i<\frac{1}{2}T<\lo$, this worker cannot have a break in between these days. In other words $n_1^*=n_2^*$ is on duty on all $a_i$ days.
        Furthermore, every worker works exactly two days before and two days after $a_i$. The described schedule is indeed feasible and thus the claim holds.
    \end{claimproof}
    Finally we set $\Uw=3m\cdot4+T$ and the statement follows by concatenating the encodings of the $a_i$ and merging trailing and leading zeros.
\end{proof}

\begin{theorem}\label{thm: complexity one sided}
    If the requests are bounded only by one side, the theoretical complexity of the \dodosp is as given in \Cref{tab: NP-hardness cases if request bounded by one side}. 
\end{theorem}
\begin{proof}
    Note that a lower bound on the requests is an upper bound on the number of workers who do not work on a given day. Therefore, if only $\requestl{d}$ is given, we receive the same results as in the case that only $\requestu{d}$ is given, but with $w$ and $o$ swapped. 

    In \Cref{sec: upper bounds,sec: local bounds} we present polynomial algorithms for the cases that only upper bounds or only local bounds are given. Since $\requestl{d}$ can be chosen arbitrary, they also apply to this case. 
    If only $\Uw, \uo,\lo$ and $\lw$ are given, the schedule in which no worker works at all is feasible. Thus, this case is trivially in P. 
    The remaining cases are considered in \Cref{lem: NP-hard rud Uw uo lw,lem: NP-hard Uo l,lem: NP-complete Uw uo lo}.
\end{proof}

Last, we remark that we can also conclude the complexity of the \cdodosp.
\begin{corollary}\label{cor: NP-hard cyclic}
    The \cdodosp is NP-complete, even if $\requestl{d} = \requestu{d}$ for all $d\in \CD$ and if there is only one nontrivial global bound and one nontrivial local bound. 
\end{corollary}
\begin{proof}
    In the proof of \Cref{thm: NP-hard in general}, the final day of the scheduling period has request 0 or $N$ depending on whether $\lw$ or $\lo$ is nontrivial. Thus, in the cyclic version the two periods representing $a_1$ and $a_{3m}$ are also separated. Therefore, we can apply the same proof.
\end{proof}

\subsection{Compuatation of feasible potential}

In this subsection we want to describe the network that is used to compute a feasible potential for the \udodosp and the \ldodosp in more detail.

\begin{proposition}\label{prp: solve udodosp inequalities}
    The integral inequality system described in \Cref{cor: udodosp feasibility} can be solved in $\runtime{D^2}$. 
\end{proposition}
\begin{proof}
   Consider the graph $G=\lr{V,E,c}$ with
    \begin{align*}
        V=&\lrb{v^d\,:\,0\leq d\leq D},\\
        E=&\lrb{e_{\Uw}=\lr{v^0,v^D},e_{\Uo}=\lr{v^D,v^0}}\\
            &\cup\lrb{e_{\uw}^d=\lr{v^{d-1},v^{d+\uw}}\,:\,1\leq d\leq D-\uw}\\
            &\cup\lrb{e_{\uo}^d=\lr{v^{d+\uo},v^{d-1}}\,:\,1\leq d\leq D-\uo}\\
            &\cup \lrb{e_{\requestl{}}^d=\lr{v^d,v^{d-1}},
                       e_{\requestu{}}^d=\lr{v^{d-1},v^d}\,:\,1\leq d\leq D}
    \end{align*}
    \vspace{-1em}
    and
    \vspace{-1em}
    \begin{align*}
        c\lr{e_{\Uw}} &= N\cdot\Uw,& c\lr{e_{\Uo}}&=N\cdot\lr{\Uo-D},\\
        c\lr{e_{\uw}^d}&=N\cdot\uw,& c\lr{e_{\uo}^d}&=-N,\\
        c\lr{e_{\requestl{}}^d}&=-\requestl{d},& c\lr{e_{\requestu{}}^d}&=\requestu{d}.
    \end{align*}
    An example of such a graph is given in \Cref{fig: upper bounds graph}. 
    \begin{figure}[htbp]
        \centering
        \includestandalone[width=0.7\textwidth]{upper_bounds_graph}
        \caption{The graph which results from requests $\lrb{1,3},\lrb{1,1},\lrb{1,4},\lrb{2,3},\lrb{4,4},\lrb{1,3},\lrb{2,4},\lrb{2,2},\lrb{1,2}$ and bounds $u_w = 4, u_o = 2, U_w = 6$ and $U_o = 4$ is shown in black. For $N = 4$ one can check that the blue values in the vertices form a feasible potential and thus imply feasibility of the \dodosp instance. This potential corresponds to the demands given in \Cref{fig: representatives algo}.}
        \label{fig: upper bounds graph}
    \end{figure}

    Recall that a function $\pi$ on the vertices of a digraph is called a feasible potential, if $\pi(v) + c(v,w) - \pi(w) \geq 0$ for each edge $\lr{v,w}$.
    The values $\Work{d}$ fulfill the inequality system given in \Cref{cor: udodosp feasibility} if and only if $\pi\lr{v^d}\coloneqq\Work{d}$ for $0\leq d\leq D$ is a feasible potential of $G$. Each edge of the graph corresponds to one inequality.
    Since the number of vertices and the number of edges in $G$ is linear in $D$, the Moore-Bellman-Ford algorithm \parencite{bellmanRoutingProblem1958,fordNetworkFlowTheory1956,mooreShortestPathMaze1959} can compute a feasible integral potential or detect a negative cycle in $\runtime{D^2}$.
\end{proof}

\end{document}

%% file: prob_intro_feasible_example.tex
    \begin{tikzpicture}[scale=1]
    \begin{scope}[opacity=0.3]
        \nurseplan{1}{0,1,2,5,6,7,8,9}{3,4,10,11,12}{$n_1$}
        \nurseplan{2}{3,4,5,6,7,10,11,12}{0,1,2,8,9}{$n_2$}
        \nurseplan{5}{0,1,2,3,4,9,10,11,12}{5,6,7,8}{$n_5$}
        \nurseplan{4}{0,1,2,3,4,8,9,10,11,12}{5,6,7}{$n_4$}
        \nurseplan{7}{4,5,6,7,8}{0,1,2,3,9,10,11,12}{$n_7$}
       
    \end{scope}
    \nurseplan{3}{2,3,4,5,10,11,12}{0,1,6,7,8,9}{$n_3$}
    \nurseplan{6}{3,4,5,6,9,10,11,12}{0,1,2,7,8}{$n_6$}
        
    \demandvalues{3,3,4,5,6,5,4,3,3,4,5,5,5}

    \draw[decoration={brace,raise=0.03cm},decorate, ultra thick, color=blue]
  (2,-2) -- node[above] {$l_w\leq\dots\leq u_w$} (6,-2);
    \draw[decoration={brace,raise=0.03cm},decorate, ultra thick, color=ForestGreen]
  (7,-5) -- node[above] {$l_o\leq\dots\leq u_o$} (9,-5);

    \draw[<-, blue, ultra thick] (4.1,-3.1) to[out=-45, in=180] ++ (1.5,-0.5) node[right] {\textbf{work period}};
    \draw[<-, ForestGreen, ultra thick] (8.1,-6.1) to[out=-135, in=0] ++ (-1.5,-0.5) node[left] {\textbf{off period}};

    \draw[blue, dashed, ultra thick] (2,-2) rectangle (6,-3);
    \draw[blue, dashed, ultra thick] (10,-2) rectangle (13,-3);
    \draw[ForestGreen, dashed, ultra thick] (7,-5) rectangle (9,-6);
    \draw[ForestGreen, dashed, ultra thick] (0,-5) rectangle (3,-6);

    \node at (13.8,-2.5) {\textcolor{blue}{$\leq U_w$}};
    \node at (13.8,-5.5) {\textcolor{ForestGreen}{$\leq U_o$}};
    \end{tikzpicture}

%% file: motivate_udodosp.tex
    \newcommand{\tempnurseplan}[4]{
        \begin{scope}[shift = {(0,-#1)}]
            \foreach \x in {#2}{
                \worksquare{\x,0}{\x+1,1}
            }
            \foreach \x in {#3}{
                \offsquare{\x,0}{\x+1,1}
            }
            \node at (1.5,0.5) {$\dots$};
            \node at (6.5,0.5) {$\dots$};
            \node at (-0.7,0.5) {#4};
        \end{scope}
    }

    \begin{tikzpicture}[scale=1]
        \demandvalues{,\dots,,1,1,,\dots,}
        \tempnurseplan{1}{3}{}{$n_1$}
        \undecidedsquare{4,-1}{5,0}
        \tempnurseplan{2}{}{3}{$n_2$}
        \undecidedsquare{4,-2}{5,-1}
        \draw [decorate, decoration={brace}, ultra thick] (3,-2) -- (0,-2) node[midway, below=5pt, align = center] {Both workers are on\\ duty equally often\\ in this interval};
        \node at (1.5,-1) {fixed schedule};
        \node at (6.5,-1) {to be scheduled};

    \end{tikzpicture}

%% file: motivate_ldodosp.tex
    \newcommand{\tempnurseplan}[4]{
        \begin{scope}[shift = {(0,-#1)}]
            \foreach \x in {#2}{
                \worksquare{\x,0}{\x+1,1}
            }
            \foreach \x in {#3}{
                \offsquare{\x,0}{\x+1,1}
            }
            \node at (1.5,0.5) {$\dots$};
            \node at (11.5,0.5) {$\dots$};
            \node at (-0.7,0.5) {#4};
        \end{scope}
    }

    \begin{tikzpicture}[scale=1]
        \demandvalues{,\dots,,0,1,2,2,1,1,0,,\dots,}
        \tempnurseplan{1}{4,5,6,7,8}{3,9}{$n_1$}
        \tempnurseplan{2}{5,6}{3,4,7,8,9}{$n_2$}
        \draw [decorate, decoration={brace}, ultra thick] (13,-2) -- (7,-2) node[midway, below=5pt, align = center] {Swap the schedule of $n_1$\\ with the schedule of $n_2$};
        \node at (1.5,-1) {fixed schedule};
        \node at (11.5,-1) {fixed schedule};

    \end{tikzpicture}

%% file: diff_instance.tex
    \begin{tikzpicture}[scale=1]
        \nurseplan{1}{2,3,7,8}{0,1,4,5,6,9,10}{$n_1$}
        \nurseplan{2}{3,4,5,8,9,10}{0,1,2,6,7}{$n_2$}
        \demandvalues{0,0,1,2,1,1,0,1,2,1,1}
    \end{tikzpicture}

%% file: no_fifo.tex
    \begin{tikzpicture}[scale=1]
        \nurseplan{1}{2,3,8,9,10}{0,1,4,5,6,7}{$n_1$}
        \nurseplan{2}{3,4,5,7,8}{0,1,2,6,9,10}{$n_2$}
        \demandvalues{0,0,1,2,1,1,0,1,2,1,1}
    \end{tikzpicture}

%% file: np_hardness_3_partition_2.tex
    \huge
    \begin{tikzpicture}[scale=1]
    \nurseplan{1}{0,1,2,8,9,10,22,23,24,25,26}{3,4,5,6,7,11,12,13,14,15,16,17,18,19,20,21,27}{$n_1$}
    \nurseplan{2}{4,5,6,12,13,14,15,17,18,19,20}{0,1,2,3,7,8,9,10,11,16,21,22,23,24,25,26,27}{$n_2$}
    \demandvalues{1,1,1,0,1,1,1,0,1,1,1,0,1,1,1,1,0,1,1,1,1,0,1,1,1,1,1,0}
    \end{tikzpicture}

%% file: representatives_example.tex
    \Large
    \begin{tikzpicture}[scale=1]

    \begin{scope}
        \nurseplan{1}{0}{}{$n_1$}
        \nurseplan{2}{0}{}{$n_2$}
        \nurseplan{3}{1}{}{$n_3$}
        \nurseplan{4}{2}{}{$n_4$}
        \nurseplan{5}{2}{}{$n_5=n_1$}
        \nurseplan{6}{3}{}{$n_6=n_2$}
        \nurseplan{7}{3}{}{$n_7=n_3$}
        \nurseplan{8}{4}{}{$n_8=n_4$}
        \nurseplan{9}{4}{}{$n_9=n_1$}
        \nurseplan{10}{4}{}{$n_{10}=n_2$}
        \nurseplan{11}{4}{}{$n_{11}=n_3$}
        \nurseplan{12}{}{}{$\vdots$}
        \node at (5.5,-11.5) {$\ddots$};
        
    \demandvalues{2,1,2,2,4,3,4,2,1}
    \end{scope}

    \begin{scope}[shift={(10,-3)}]
        \nurseplan{1}{0,2,4,5,6,8}{1,3,7}{$n_1$}
        \nurseplan{2}{0,3,4,5,6}{1,2,7,8}{$n_2$}
        \nurseplan{3}{1,3,4,6,7}{0,2,5,8}{$n_3$}
        \nurseplan{4}{2,4,5,6,7}{0,1,3,8}{$n_4$}
        \demandvalues{2,1,2,2,4,3,4,2,1}
    \end{scope}

    \draw[decoration={brace,raise=0.03cm},decorate, ultra thick, color=blue]
  (1,-0) -- node[right] {$r^{1}=2$} (1,-2);
  \draw[decoration={brace,raise=0.03cm},decorate, ultra thick, color=blue]
  (2,-2) -- node[right] {$r^{2}=1$} (2,-3);
  \draw[decoration={brace,raise=0.03cm},decorate, ultra thick, color=blue]
  (3,-3) -- node[right] {$r^{3}=2$} (3,-5);
  \draw[decoration={brace,raise=0.03cm},decorate, ultra thick, color=blue]
  (4,-5) -- node[right] {$r^{4}=2$} (4,-7);
  \draw[decoration={brace,raise=0.03cm},decorate, ultra thick, color=blue]
  (5,-7) -- node[right] {$r^{5}=4$} (5,-11);

    \draw[blue,decorate, decoration={snake, segment length = 0.5cm},->, ultra thick] (7,-5) -- (8,-5);

    \end{tikzpicture}

%% file: uw_violation.tex
    \newcommand{\tempnurseplan}[4]{
        \begin{scope}[shift = {(0,-#1)}]
            \foreach \x in {#2}{
                \worksquare{\x,0}{\x+1,1}
            }
            \foreach \x in {#3}{
                \offsquare{\x,0}{\x+1,1}
            }
            \node at (-0.5,0.5) {$\dots$};
            \node at (7.5,0.5) {$\dots$};
            \node at (-1.7,0.5) {#4};
        \end{scope}
    }

    \begin{tikzpicture}[scale=1]
        \tempnurseplan{1}{1,2,3,4,5,6}{0}{$n_1$}
        \tempnurseplan{2}{0,1,2,3,4,5,6}{}{$n_2$}
        \tempnurseplan{3}{0,1,2,3,4,5}{6}{$n_3$}
        \tempnurseplan{4}{0,1,2,3,4,5}{6}{$n_4$}
        \node at (0.5,-4.5) {$d_1$};
        \node at (1.5,-4.5) {$d_2$};
        \node at (2.5,-4.5) {$\dots$};
        \node at (6.5,-4.5) {$d_{u_w+1}$};

    \end{tikzpicture}

%% file: period_counters.tex
    
\begin{tikzpicture}[scale=1]
    \nurseplan{1}{0,1,5,6,7,12,13}{2,3,4,8,9,10,11}{$n_1$}
    \nurseplan{2}{0,1,2,6,8,7,12,13}{3,4,5,9,10,11}{$n_2$}
    \nurseplan{3}{0,1,2,6,7,8,9}{3,4,5,10,11,12,13}{$n_3$}
    \nurseplan{4}{0,1,2,3,7,8,9,10}{4,5,6,11,12,13}{$n_4$}
    \nurseplan{5}{2,3,4,8,9,10,11}{0,1,5,6,7,12,13}{$n_5$}
    \nurseplan{6}{4,5,6,10,11,12,13}{0,1,2,3,7,8,9}{$n_6$}
    
    \demandvalues{4,4,4,2,2,2,4,4,4,3,3,2,3,3}
    \orderedperiodstarts{6}{0,0,0,0,2,4,5,6,6,7,8,10,12,12}
    \orderedperiodends{6}{1,2,2,3,4,6,7,8,9,10,11,13,13,13}
   
    \begin{scope}[green!60!black]
        \periodcounter{4,4,5,5,6,7,9,10,11,11,12,12,14,14}{7}{$S_w^d$}
    \end{scope}
    \begin{scope}[red]
        \periodcounter{0,0,1,3,4,5,5,6,7,8,9,10,11,11}{8}{$T_w^d$}
    \end{scope}
    \begin{scope}[gray]
        \periodcounter{2,2,3,5,6,7,7,8,9,10,11,12,13,13}{9}{$S_o^d$}
    \end{scope}
    \begin{scope}[gray]
        \periodcounter{0,0,1,1,2,3,5,6,7,7,8,8,10,10}{10}{$T_o^d$}
    \end{scope}


\end{tikzpicture}

%% file: np_hardness_sequence_special_case.tex
\begin{tikzpicture}
    \def\vdis{0.8}
    \setcounter{mycounter}{0} 


    \draw[dashed, thick, gray] (12*\vdis,-0.3) -- (12*\vdis, 0.8);
    \draw[<-, gray, thick] (12*\vdis + 0.1, 0.6) to[out=0, in=180] ++ (1,0) node[right] {line of symmetry};
    
    \addtriple{0};
    \addbox{N};
    \addbox{1}
    \addbox{0}
    \addbox{N\!\!-\!\!1}
    \addtriple{1}
    \add{}
    \add{\cdots}

    \draw [decorate, decoration={brace,mirror}] (-0.2*\vdis,-0.3) -- (2.2*\vdis,-0.3) node[midway, below=5pt] {$u_o$};
    \draw [decorate, decoration={brace,mirror}] (10.8*\vdis,-0.3) -- (13.2*\vdis,-0.3) node[midway, below=5pt] {$a_i$};
    
\end{tikzpicture}

%% file: np_hardness_sequence_2.tex
\begin{tikzpicture}
    \def\vdis{0.7}

    \draw[dashed, thick, gray] (13*\vdis,-0.3) -- (13*\vdis, 0.8);
    \draw[<-, gray, thick] (13*\vdis + 0.1, 0.6) to[out=0, in=180] ++ (0.7,0) node[right] {line of symmetry};

    \setcounter{mycounter}{0}
    \addtriple{0};
    \add{N}
    \add{1}
    \addtriple{0}
    \add{N\!\!-\!\!1}
    \addtriple{0}
    \addtriple{1}
    
    \draw [decorate, decoration={brace,mirror}] (-0.2*\vdis,-0.3) -- (2.2*\vdis,-0.3) node[midway, below=5pt] {$u_o=\frac{3}{2}T+1$};
    \draw [decorate, decoration={brace,mirror}] (4.8*\vdis,-0.3) -- (11.2*\vdis,-0.3) node[midway, below=5pt] {$u_o=\frac{3}{2}T+1$};
    \draw [decorate, decoration={brace}] (4.8*\vdis,0.3) -- (7.2*\vdis,0.3) node[midway, above=5pt] {$T$};
    \draw [decorate, decoration={brace}] (8.8*\vdis,0.3) -- (11.2*\vdis,0.3) node[midway, above=5pt] {$\frac{1}{2}T$};
    
    \draw [decorate, decoration={brace,mirror}] (11.8*\vdis,-0.3) -- (14.2*\vdis,-0.3) node[midway, below=5pt] {$a_i$};
    
\end{tikzpicture}

%% file: upper_bounds_graph.tex
    \normalsize
    \begin{tikzpicture}[every node/.style=edgenode,scale=1]
        \Large
        \node[circlenode] (V0) at ( 0,0) {$0$};
        \node[circlenode] (V1) at ( 2,0) {$2$};
        \node[circlenode] (V2) at ( 4,0) {$3$};
        \node[circlenode] (V3) at ( 6,0) {$5$};
        \node[circlenode] (V4) at ( 8,0) {$7$};
        \node[circlenode] (V5) at (10,0) {$11$};
        \node[circlenode] (V6) at (12,0) {$14$};
        \node[circlenode] (V7) at (14,0) {$18$};
        \node[circlenode] (V8) at (16,0) {$20$};
        \node[circlenode] (V9) at (18,0) {$21$};

        \node[varnode] at (-1.5,0.05) {$W^0=$};
        \node[varnode,text width=15mm] at (19.1,0.05) {$=W^9$};
        \normalsize

        \path[normaledge] (V0) edge[bend left = 20] node {$3$} (V1);
        \path[normaledge] (V1) edge[bend left = 20] node {$1$} (V2);
        \path[normaledge] (V2) edge[bend left = 20] node {$4$} (V3);
        \path[normaledge] (V3) edge[bend left = 20] node {$3$} (V4);
        \path[normaledge] (V4) edge[bend left = 20] node {$4$} (V5);
        \path[normaledge] (V5) edge[bend left = 20] node {$3$} (V6);
        \path[normaledge] (V6) edge[bend left = 20] node {$4$} (V7);
        \path[normaledge] (V7) edge[bend left = 20] node {$2$} (V8);
        \path[normaledge] (V8) edge[bend left = 20] node {$2$} (V9);
        
        \path[normaledge] (V1) edge[bend left = 20] node {$-1$} (V0);
        \path[normaledge] (V2) edge[bend left = 20] node {$-1$} (V1);
        \path[normaledge] (V3) edge[bend left = 20] node {$-1$} (V2);
        \path[normaledge] (V4) edge[bend left = 20] node {$-2$} (V3);
        \path[normaledge] (V5) edge[bend left = 20] node {$-4$} (V4);
        \path[normaledge] (V6) edge[bend left = 20] node {$-1$} (V5);
        \path[normaledge] (V7) edge[bend left = 20] node {$-2$} (V6);
        \path[normaledge] (V8) edge[bend left = 20] node {$-2$} (V7);
        \path[normaledge] (V9) edge[bend left = 20] node {$-1$} (V8);

        \path[normaledge] (V0) edge[bend left = 40] node {$4N$} (V5);
        \path[normaledge] (V1) edge[bend left = 40] node {$4N$} (V6);
        \path[normaledge] (V2) edge[bend left = 40] node {$4N$} (V7);
        \path[normaledge] (V3) edge[bend left = 40] node {$4N$} (V8);
        \path[normaledge] (V4) edge[bend left = 40] node {$4N$} (V9);

        \path[normaledge] (V3) edge[bend left = 40] node {$-N$} (V0);
        \path[normaledge] (V4) edge[bend left = 40] node {$-N$} (V1);
        \path[normaledge] (V5) edge[bend left = 40] node {$-N$} (V2);
        \path[normaledge] (V6) edge[bend left = 40] node {$-N$} (V3);
        \path[normaledge] (V7) edge[bend left = 40] node {$-N$} (V4);
        \path[normaledge] (V8) edge[bend left = 40] node {$-N$} (V5);
        \path[normaledge] (V9) edge[bend left = 40] node {$-N$} (V6);

        \path[normaledge] (V0) edge[bend left = 40] node {$6N$} (V9);
        \path[normaledge] (V9) edge[bend left = 40] node {$-5N$} (V0);
    \end{tikzpicture}